\documentclass[12pt]{article}

\usepackage{amsmath,amssymb,amsthm}
\usepackage{geometry}
\usepackage{tikz,tkz-graph,tikz-cd,tikz-qtree,youngtab}

\newtheorem{theorem}{Theorem}[section]
\newtheorem{lemma}[theorem]{Lemma}
\newtheorem{corollary}[theorem]{Corollary}

\newtheorem{proposition}[theorem]{Proposition}

\theoremstyle{remark}
\newtheorem{remark}[theorem]{Remark}
\theoremstyle{definition}
\newtheorem{example}[theorem]{Example}
\theoremstyle{definition}

\begin{document}
	
\bibliographystyle{abbrv}
	
\title{BCH and LCD cyclic codes of length $n=\lambda(q^m+1)$ over finite fields}
	
\author{Jinle Liu$^{1}$, Hongfeng Wu$^{1}$ and Li Zhu$^{2}$ \footnote{Corresponding author.}
\setcounter{footnote}{-1}
\footnote{E-Mail addresses:
cohomologyliu@163.com (J. Liu), whfmath@gmail.com (H. Wu), lizhumath@pku.edu.cn (L. Zhu)}
\\
{1.~College of Science, North China University of technology, Beijing, China}
\\
{2.~School of Mathematical Sciences, Guizhou Normal University, Guiyang, China}}
	
\date{}
\maketitle

\thispagestyle{plain}
\setcounter{page}{1}	

\begin{abstract}
BCH and LCD cyclic codes of length $n=\lambda(q^m+1)$ with $\lambda\mid q-1$ are studied. A complete characterization of $q$-cyclotomic cosets modulo $n$ is given: Theorem \ref{th4} provides a necessary and sufficient condition for any $0\le \gamma<n$ to be a coset leader, and for odd $m$, the two largest coset leaders are explicitly determined (Theorem \ref{th9} and Theorem \ref{th14}). Based on these results, the dimensions of several families of BCH codes are determined, and the lower bound on the minimal distance of $\mathcal{C}_{(q,n,2\delta+1,n-\delta+1)}$ is raised to $2(\delta+1)$ (Theorem \ref{th15}--\ref{th5}). Notably, several of these codes are optimal. When $m$ is odd, the necessary and sufficient condition for the BCH code $\mathcal{C}_{(q,n,\delta,0)}$ to be dually-BCH is proved (Theorem \ref{th11}). Finally, an exact enumeration of all LCD cyclic codes of this length is derived (Theorem \ref{th3}). All of the above results extend previous results that were limited to $\lambda=1$.


\end{abstract}

\section{Introduction}
Cyclic codes, a prominent subclass of error-correcting codes, have been extensively studied and are widely used in practice. 
Their rich algebraic structure enables efficient encoding and decoding, making them indispensable in numerous fields such as 
communications, data storage, and aerospace engineering. Important families of cyclic codes (and their extended 
variants)-including BCH codes, Reed-Solomon codes, Hamming codes, and Reed-Muller codes-form the foundation of modern coding 
theory. A central problem in the study of cyclic codes is the determination of two key parameters: dimension and minimum distance, 
which directly govern a code’s error-correcting capability and practical utility.

BCH codes, as an important class of cyclic codes, are widely employed in satellite communications and other critical systems. 
They were first introduced by Hocquenghem \cite{hoc} and independently developed by Bose and Ray-Chaudhuri \cite{bose}. 
BCH codes offer two key advantages: precise control over the number of correctable symbol errors, and efficient encoding 
and decoding algorithms. These properties have led to their adoption in diverse devices and systems, including DVDs, solid-state drives, disk drives, compact disc players, and satellite networks.

Primitive BCH codes(length $q^m-1$) over $\mathbb{F}_q$ have been 
studied most intensively, with notable contributions found in, e.g., \cite{aly, cherchem, dingcunsheng2, dingcunsheng4, 
gongbinkai, huangxinmei, lichengju, dingcunsheng7, dingcunsheng6, yuming, yuedianwu, zhangyanhui}. In particular, 
the well-known Reed–Solomon codes are BCH codes of length $q-1$. Beyond the primitive case, many results are also available for 
BCH codes of non-primitive lengths. BCH codes of length $\frac{q^m-1}{q-1}$ are also of great interest. Notably, many Hamming 
codes are narrow-sense BCH codes of this length. In general, it is difficult to determine the dimension, minimum distance, 
weight distribution and other parameters of BCH codes. The parameters of BCH codes of length $\frac{q-1}{\lambda}$ for 
$\lambda=q-1$ and $\lambda\mid q-1$ were investigated in \cite{lichengju, dingcunsheng7} and \cite{zhu}. Codes of 
length $n=q^m+1$ were studied in \cite{fuyuqing, lichengju, dingcunsheng7, yanhaode}; and codes of length 
$\frac{2(q^{2m}-1)}{q+1}$ and $\frac{q^{2m}-1}{q+1}$ were considered in 
\cite{pangbinbin, zhuhuan}. This list, though not exhaustive, reflects the breadth of known work in the area.

A linear code $\mathcal{C}$ over $\mathbb{F}_q$ is called linear complementary dual code(LCD code) provided $\mathcal{C}\cap 
\mathcal{C}^{\perp}=\{\textbf{0}\}$.  In earlier studies, cyclic LCD codes over finite fields were sometimes called 
reversible codes and were initially investigated by Massey for applications in data storage systems \cite{mass}. LCD codes 
have attracted considerable attention in recent years because of their relevance in cryptography, where they play an 
important role in protecting against side-channel attacks and fault non-invasive attacks (see \cite{bri, cle} for 
comprehensive discussions). A significant result from \cite{lichengju} states that a cyclic code of length $n$ 
over $\mathbb{F}_q$ is LCD whenever $-1$ can be written as a power of $q$ modulo $n$.

It is well known that the parameters of BCH codes and LCD cyclic codes are closely related to the structure of $q$-cyclotomic cosets. Precise descriptions of these cyclotomic cosets—including their representatives, leaders, sizes, and enumeration—are essential for determining code dimension and minimum distance. Considerable progress has been made in special cases: for example, the representatives and sizes of various families of cyclotomic cosets were determined in \cite{Chen}, \cite{Chen 2}, \cite{Geraci}, \cite{Sharma}, \cite{Liu}, \cite{Sharma 2}, \cite{Liu 2}, \cite{Wu 2},  \cite{Yue}, etc.. Recently, for more 
general settings, Zhu et al. 
\cite{zhuli1, zhuli2}  gave an explicit description of the representatives and sizes of all $q$-cyclotomic cosets modulo $n$. 
In further work, they \cite{zhuli3} introduced the notion of equal-difference cyclotomic cosets, showing that every 
cyclotomic coset can be decomposed into a disjoint union of equal-difference subsets, and within this framework 
they defined the multiple equal-difference property of $q$-cyclotomic cosets.

As far as we know, for $n=\lambda(q^m+1)$, where $\lambda\mid q-1$, only a few papers have investigated BCH codes of 
length $n$ in the case $\lambda=1$. This is primarily due to the extreme complexity of the structure of the $q$-cyclotomic 
cosets modulo $n$. \cite{lichengju} determined the parameters of $\mathcal{C}_{(q,n,\delta,0)}$ for $3\le \delta\le q^{\lfloor\frac{m-1}{2}\rfloor}+3$. Specially, they improved the BCH bound on their minimum distance. Subsequently, \cite{dingcunsheng7} determined the dimensions of narrow-sense BCH code $\mathcal{C}_{(q,n,\delta,1)}$, 
where $2\le \delta\le q^{\frac{m}{2}}$ for even $m\ge 4$ and $2\le \delta\le q^{\frac{m+1}{2}}$ for odd $m\ge 3$, 
respectively. In recent years, to further explore the dual codes of BCH codes, the notion of dually-BCH codes was introduced by authors in \cite{gongbinkai}. Recently, \cite{fuyuqing} developed a lower bound on the minimum distance of the dual code 
$\mathcal{C}_{(q,n,\delta,1)}$ and gave a sufficient and necessary condition for the even-like subcode 
$\mathcal{C}_{(q,n,\delta+1,0)}$ of $\mathcal{C}_{(q,n,\delta,1)}$ being dually-BCH. For more related results, 
we refer the reader to \cite{liuy, liuy2, yanhaode, zhuhongwei}.

In this paper, we always let $n=\lambda(q^m+1)$, where $\lambda\mid q-1$. We investigate the structure of $q$-cyclotomic 
cosets modulo $n$, as well as the properties of BCH codes and LCD cyclic codes of this length. The rest of this paper is organized as follows. In Section 2, we present some preliminary notions and results that will be 
used later. In Section 3, we explore the $q$-cyclotomic cosets modulo $n$. Specifically, we establish a necessary and 
sufficient condition for an integer $0\le \gamma<n$ to be a coset leader. For odd $m$, we further determine the largest and 
second largest coset leaders. In Section 4, we determine the dimensions of several families of BCH codes and improve the 
lower bounds on their minimum distances in certain cases. In particular, some of the constructed BCH codes are shown to be 
optimal. Additionally, we present a necessary and sufficient condition for a BCH code $\mathcal{C}_{(q,n,\delta,0)}$ to be a 
dually-BCH code. In Section 5, we provide the exact enumeration of LCD cyclic codes of length $n$. Finally, Section 6 concludes 
the paper.

\section{Preliminaries}

In this section we recall some basic definitions and results in the theory of finite fields and cyclic codes over finite fields.

Let $q=p^{e}$ be a power of a prime number $p$, and $\mathbb{F}_{q}$ be a finite field containing $q$ elements. Let $n$ be a positive integer not divided by $p$, and $\zeta_{n}$ be a fixed primitive $n$-th root of unity which lies in some finite extension of $\mathbb{F}_{q}$. It is a well-known fact that the irreducible factors of $x^{n}-1$ over $\mathbb{F}_{q}$ are in an one-to-one correspondence with the $q$-cyclotomic cosets modulo $n$.

To be explicit, for any $\gamma \in \mathbb{Z}/n\mathbb{Z}$, the $q$-cyclotomic coset modulo $n$ containing $\gamma$ is defined to be 
$$C_{\gamma} = \{\gamma,\gamma q,\cdots,\gamma q^{\tau-1}\} \subseteq \mathbb{Z}/n\mathbb{Z},$$
where $\tau$ is the smallest positive integer such that $\gamma q^{\tau} = \gamma$ in $\mathbb{Z}/n\mathbb{Z}$, and is called the size of $C_{\gamma}$. Viewing the elements in $C_{\gamma}$ as integers in the range $\{0,1,\cdots,n-1\}$, the smallest one is called the coset leader of $C_{\gamma}$. We denote by $\Gamma = \Gamma_{(n,q)}$ the set of all coset leaders of $q$-cyclotomic cosets modulo $n$. Any $q$-cyclotomic coset modulo $n$ induces an irreducible factor of $x^{n}-1$ via
$$f_{\gamma}(x) = (x-\zeta_{n}^{\gamma})(x-\zeta_{n}^{\gamma q})\cdots(x-\zeta_{n}^{\gamma q^{\tau-1}}).$$
Moreover, all the irreducible factors of $x^{n}-1$ over $\mathbb{F}_{q}$ can be obtained in this way, that is, the irreducible factorization of $x^{n}-1$ over $\mathbb{F}_{q}$ is given by
$$x^{n}-1 = \prod_{\gamma \in \Gamma_{n,q}}f_{\gamma}(x).$$

Let $f(x) = x^{r} + c_{r-1}x^{r-1}+ \cdots +c_{0}$ be a polynomial over $\mathbb{F}_{q}$ with $c_{0} \neq 0$. The reciprocal polynomial $f^{\ast}(x)$ of $f(x)$ is defined to be
$$f^{\ast}(x) = c_{0}^{-1}\cdot x\cdot f(x^{-1}).$$
The polynomial $f(x)$ (or equivalently, $f^{\ast}(x)$) is called self-reciprocal if $f(x) = f^{\ast}(x)$. The following lemmas are classical results on self-reciprocal polynomials.

\begin{lemma}\cite{lichengju}
	The irreducible polynomial $f_{\gamma}(x)$ is self-reciprocal if and only if $n-\gamma\in C_{\gamma}$.
\end{lemma}

\begin{lemma}\cite{lichengju}
	The least common multiple $\mathrm{lcm}(f_{\gamma}(x), f_{n-\gamma}(x))$ is self-reciprocal for every $\gamma\in \mathbb{Z}/n\mathbb{Z}$.
\end{lemma}

A linear code $\mathcal{C}$ of length $n$ over $\mathbb{F}_{q}$ is a subspace of $\mathbb{F}_{q}^{n}$. If $\mathcal{C}$ satisfies the property that for every $(c_{0},c_{1},\cdots,c_{n-1}) \in \mathcal{C}$ the shifted vector $(c_{n-1},c_{0},\cdots,c_{n-2})$ is also contained in $\mathcal{C}$, then $\mathcal{C}$ is called a cyclic code. Let $\mathcal{R}_{n} = \mathbb{F}_{q}[x]/(x^{n}-1)$. Then $\mathbb{F}_{q}^{n}$ is isomorphic to $\mathcal{R}_{n}$ as $\mathbb{F}_{q}$-spaces, via the map
$$(c_{0},c_{1},\cdots,c_{n-1}) \mapsto c_{0} + c_{1}x + \cdots + c_{n-1}x^{n-1}.$$
Under this isomorphism, a linear code of length $n$ over $\mathbb{F}_{q}$ is cyclic if and only if it corresponds to an ideal $(f(x)) \subseteq \mathcal{R}_{n}$, where $f(x)$ is a factor of $x^{n}-1$. If it is this case, then we identify the code $\mathcal{C}$ with the corresponding ideal $(f(x))$, and call $f(x)$ the generator polynomial of $\mathcal{C}$.

Let $\mathcal{C}$ be a linear code of length $n$ over $\mathbb{F}_{q}$. The dual code $\mathcal{C}^{\perp}$ of $\mathcal{C}$ is defined as
$$\mathcal{C}^{\perp} = \{\boldsymbol{x} \in \mathbb{F}_{q}^{n} \ | \ \boldsymbol{x}\cdot \boldsymbol{c}^{T} = 0 \ \mathrm{for} \ \forall \boldsymbol{c} \in \mathcal{C}\},$$
where $\boldsymbol{x}\cdot \boldsymbol{c}^{T}$ denotes the standard inner product of $\boldsymbol{x},\boldsymbol{c} \in \mathbb{F}_{q}$. The code $\mathcal{C}$ is called a LCD code or a reversible code if 
$$\mathcal{C} \cap \mathcal{C}^{\perp} = \{\boldsymbol{0}\}.$$
If $\mathcal{C} = (f(x))$ is further a cyclic code where $f(x) \mid x^{n}-1$, then its dual code $\mathcal{C}^{\perp}$ is also a cyclic code of length $n$, with generator polynomial $g^{\ast}(x)$, where $g(x) = \frac{x^{n}-1}{f(x)}$ and $g^{\ast}(x)$ is the reciprocal polynomial of $g(x)$. The next lemma gives criteria for a cyclic code being reversible.

\begin{lemma}\cite{Yang2}
	Let $\mathcal{C}$ be a cyclic code of length $n$ over $\mathbb{F}_q$ with generator polynomial $f(x)$. Then the following statements are equivalent.
	\begin{itemize}
		\item[(1)] The code $\mathcal{C}$ is a LCD code.
		\item[(2)] The polynomial $f(x)$ is self-reciprocal.
		\item[(3)] For every root $\alpha$ of $f(x)$, $\alpha^{-1}$ is also a root of $f(x)$.
	\end{itemize}
	In particular, if $-1$ is congruent to a power of $q$ modulo $n$, then every cyclic code over $\mathbb{F}_q$ of length $n$ is reversible.
\end{lemma}
	
BCH codes are a subclass of cyclic codes, which are significant both in the theoretical and the practical aspects. We briefly recall the definition of BCH code and the BCH bound. Let $\mathcal{C}$ be a cyclic code of length $n$ with generator polynomial $f(x)$. Fixing a primitive $n$-th root of unity lying in some finite extension of $\mathbb{F}_{q}$, then $\mathcal{C}$ is uniquely determined by the set
$$T = \{0 \leq i \leq n-1 \ | \ f(\zeta_{n}^{i})=0\},$$
which is referred to as the defining set of $\mathcal{C}$ with respect to $\zeta_{n}$. Let $\delta$ be an integer in the range $2 \leq \delta \leq n-1$ and $b$ be an integer. The cyclic code determined by the defining set
$$T = C_{b} \cup C_{b+1} \cup \cdots \cup C_{b+\delta-2}$$
is called a BCH code of length $n$ and designed distance $\delta$, and is denoted by $\mathcal{C}_{(q,n,\delta,b)}$. Equivalently, the BCH code $\mathcal{C}_{(q,n,\delta,b)}$ has generator polynomial 
$$f_{(q,n,\delta,b)}(x) = \mathrm{lcm}(f_{b}(x),f_{b+1}(x),\cdots,f_{b+\delta-2}(x)).$$
It may happen that $\mathcal{C}_{(q,n,\delta,b)} = \mathcal{C}_{(q,n,\delta^{\prime},b)}$ for some $\delta^{\prime}\ne \delta$. The maximal integer $d_{B}$ such that $\mathcal{C}_{(q,n,\delta,b)} = \mathcal{C}_{(q,n,d_{B},b)}$ is called the Bose distance of $\mathcal{C}_{(q,n,\delta,b)}$. When $b=1$, the code $\mathcal{C}_{(q,n,\delta,b)}$ is referred to as a narrow-sense BCH code. If $n=q^{m}-1$, then $\mathcal{C}_{(q,n,\delta,b)}$ is called a primitive BCH code; if $n = q^{m}+1$, then $\mathcal{C}_{(q,n,\delta,b)}$ is called an antiprimitive BCH code.

Let $\mathcal{C}_{(q,n,\delta,b)}$ be a BCH code, with defining set $T$ and Bose distance $d_{B}$. Then the dimension of $\mathcal{C}_{(q,n,\delta,b)}$ is $n - |T|$, and the minimum distance $d$ of $\mathcal{C}_{(q,n,\delta,b)}$ satisfies
\begin{equation}\label{eq 1}
	d \geq d_{B} \geq \delta.
\end{equation}
The bound \eqref{eq 1} is an immediate consequence of the famous BCH bound, which is stated as follow.

\begin{theorem}
	Let $\mathcal{C}$ be a cyclic code of length $n$ over $\mathbb{F}_{q}$, with defining set $T$ and minimal distance $d$. If $T$ contains $\delta-1$ consecutive integers of some integer $\delta$, then $d \geq \delta$.
\end{theorem}
	
We conclude this section by recording the lifting-the-exponent lemma, which will be needed in the following context.

\begin{lemma}\cite{Nezami}
	Let $m$ be an odd integer, $d$ be a positive integer.
	\begin{itemize}
		\item[(1)] If $m \equiv 1 \pmod{4}$, then
		$$v_{2}(m^{d}-1) = v_{2}(m-1) + v_{2}(d), \  v_{2}(m^{d}+1) = 1.$$
		\item[(2)] If $m \equiv 3 \pmod{4}$ and $d$ is odd, then
		$$v_{2}(m^{d}-1) = 1, \  v_{2}(m^{d}+1) = v_{2}(m+1).$$
		\item[(3)] If $m \equiv 3 \pmod{4}$ and $d$ is even, then
		$$v_{2}(m^{d}-1) = v_{2}(m+1) + v_{2}(d), \  v_{2}(m^{d}+1) = 1.$$
	\end{itemize}
\end{lemma}

\section{Coset leaders and sizes of $q$-cyclotomic cosets modulo $n$}

Throughout this section, $n=\lambda(q^m+1)$, where $m$ is a positive integer and $\lambda\mid q-1$. We begin by investigating the reversible properties of $q$-cyclotomic coset modulo $n$, which serves as a fundamental tool for 
subsequent analyses. We then develop a general method for determining coset leaders. Based on this, we establish a necessary and sufficient condition for $0\le \gamma<n$ being a coset leader. Finally, for odd $m$, we obtain the largest and second largest coset leaders along with their corresponding sizes. All of these results play a crucial role in the study of BCH codes of length $n$ in the next section.

The following lemma describe the reversible properties of $q$-cyclotomic cosets modulo $n$, which will be used frequently later.

\begin{lemma}\label{lem2}
For any integer $\gamma$ with $0\le \gamma< n$, write $\gamma\equiv a\pmod{\lambda}$ such that $1\le a\le \lambda$. Then the following hold:
\begin{itemize}
\item[(1)] If $0\le \gamma\le a(q^m+1)$, then $a(q^m+1)-\gamma\in  C_{\gamma}$; 
\item[(2)] If $a(q^m+1)< \gamma< n$, then $(\lambda+a)(q^m+1)-\gamma\in C_{\gamma}$.
\end{itemize}
\end{lemma}

\begin{proof}
Let $\gamma$ be an integer with $0\le \gamma<n$, and write $\gamma\equiv a\pmod{\lambda}$ such that $1\le a\le \lambda$.
\begin{itemize}
\item[(1)] Assume $0\le \gamma\le a(q^m+1)$. There exists a positive integer $m$ satisfies
\begin{align*}
a(q^m+1)-\gamma\equiv \gamma q^m\pmod{n},
\end{align*} 
which is equivalent to $a(q^m+1)-\gamma\in C_{\gamma}$.
\item[(2)] Suppose $a(q^m+1)< \gamma<n$. Similarly, we also have 
\begin{align*}
(\lambda+a)(q^m+1)-\gamma\equiv \gamma  q^{m}\pmod{n},
\end{align*}
that means $(\lambda+a)(q^m+1)-\gamma\in C_{\gamma}$.
\end{itemize}
\end{proof}

\begin{remark}
Let $\gamma$ be an integer such that $0\le \gamma<n$. It makes sense to say $\gamma\equiv a\pmod{\lambda}$ with $1\le a\le \lambda$, as it does not depend on the choice of the representatives of $q$-cyclotomic coset containing $\gamma$.
\end{remark}

Next, we propose a method for determining the coset leaders and sizes of q-cyclotomic cosets modulo n. Leveraging this approach, we establish several necessary and sufficient conditions that characterize the coset leaders and their corresponding sizes for q-cyclotomic cosets modulo n.

In fact, the size of each $q$-cyclotomic coset modulo $n$ is either $1$ or necessarity an even number. By contradiction, suppose there exists such a $q$-cyclotomic coset modulo $n$ containing $\gamma$ satisfies $\gamma\equiv a\pmod{\lambda}$ with odd size $\tau\ge 3$ such that
$$C_{\gamma}=\{\gamma, \gamma q,\cdots ,\gamma q^{\tau-1}\},$$
then according to the Lemma \ref{lem2}, there must exists an element in the set $C_{\gamma}$, we may without loss of generality denote it as $\gamma$ which satisfies $\gamma=\frac{a}{2}(q^m+1)$ or $\gamma=\frac{\lambda+a}{2}(q^m+1)$.

If $\gamma=\frac{a}{2}(q^m+1)$, for any $1\le t\le \tau-1$, we have
$$\gamma q^t\pmod{n}\equiv \frac{a}{2}\frac{q^t-1}{\lambda}\cdot \lambda(q^m+1)+\frac{a}{2}(q^m+1)\pmod{n}.$$ 

If $\gamma=\frac{\lambda+a}{2}(q^m+1)$, for any $1\le t\le \tau-1$, we have
$$\gamma q^t\pmod{n}\equiv \frac{\lambda+a}{2}\frac{q^t-1}{\lambda}\cdot \lambda(q^m+1)+\frac{\lambda+a}{2}(q^m+1)\pmod{n}.$$

It is clearly that the cyclotomic coset $C_{\gamma}$ can take three possible forms: $C_{\gamma}=\{\frac{a}{2}(q^m+1)\}$, $C_{\gamma}=\{\frac{\lambda+a}{2}(q^m+1)\}$ or $C_{\gamma}=\{\frac{a}{2}(q^m+1), \frac{\lambda+a}{2}(q^m+1)\}$. This leads to a contradiction with our assumption, thus we have proven our assertion.

For any $0\le \gamma<n$, we now assume that the size of the cyclotomic coset $C_{\gamma}$ containing $\gamma$ is $\tau$, where $\tau=1$ or $\tau$ is even.

When $\tau=1$, obviously, $\gamma$ itself is a coset leader.

When $\tau\ge 2$, we derive the coset leader of $C_{\gamma}$ based on the definition below. For any positive integer $t$ with $0\le t\le \frac{\tau}{2}-1$, write $\gamma\equiv a\pmod{\lambda}$ such that $1\le a\le \lambda$, define
\begin{equation*}
\overline{\gamma q^{t}}=	\left\{
 \begin{array}{lcl}
a(q^m+1)-(\gamma q^{t}\pmod{n}), \text{if}~\frac{a}{2} (q^m+1)<\gamma q^{t}\pmod{n}< a(q^m+1);\\
(\lambda+a)(q^m+1)-(\gamma q^{t}\pmod{n}), \text{if}~\frac{\lambda+a}{2} (q^m+1)<\gamma q^{t}\pmod{n}<n;\\
\gamma q^{t}\pmod{n}, \text{otherwise}.
 \end{array} \right.
\end{equation*}

It should be emphasized that when $a=\lambda$, the case where $\frac{\lambda+a}{2}(q^m+1)<\gamma q^{t}<n$ does not exist.

Based on the above discussion, we can express 
$$C_{\gamma}=\{\overline{\gamma},\overline{\gamma q},\cdots ,\overline{\gamma q^{\frac{\tau}{2}-1}}, {\overline{\gamma}}^{\prime},{\overline{\gamma q}}^{\prime},\cdots ,{\overline{\gamma q^{\frac{\tau}{2}-1}}}^{\prime}\},$$
where for any $0\le t\le \frac{\tau}{2}-1$,
\begin{equation*}
{\overline{\gamma q^{t}}}^{\prime}=	\left\{
 \begin{array}{lcl}
a(q^m+1)-\overline{\gamma q^{t}}, \text{if}~0\le \gamma q^{t}\pmod{n}< a(q^m+1);\\
\frac{\lambda+a}{2}(q^m+1)-\overline{\gamma q^{t}}, \text{otherwise}.
 \end{array} \right.
\end{equation*}

It is not difficult to find that $\overline{\gamma q^{t}}<\overline{\gamma q^{t}}^{\prime}$ holds for any $0\le t\le \frac{\tau}{2}-1$. Therefore, the coset leader of $C_{\gamma}$ is actually the $\mathrm{min}\{\overline{\gamma q^{t}}\mid 0\le t\le \frac{\tau}{2}-1\}$. Furthermore, for any $0\le \gamma<n$, to establish that $\gamma$ is the coset leader of $C_{\gamma}$, it suffices to show that $\overline{\gamma q^{t}}\ge \gamma$ for all $1\le t\le m-1$.

For $\lambda=1$, Lemma 15 of \cite{lichengju} and Proposition 35, Prooposition 36 of $\cite{dingcunsheng7}$ characterized all $q$-cyclotomic coset leaders modulo $n$ in the range of $1\le \gamma \le q^{\lfloor\frac{m+1}{2}\rfloor}$.  Building upon this discussion, we will consider the case for $\lambda\mid q-1$.

When $m\ge 2$ is even, we have the following conclusions.

\begin{theorem}\label{th12}
Let $m=2$. For any positive integer $\gamma$ with $1\le \gamma\le \lambda q-1$ and $\gamma\not\equiv 0\pmod{q}$, we have
\begin{itemize}
\item[(1)] If $q$ is even or $q$ is odd and $\lambda\le \frac{q-1}{2}$, then $\gamma$ is the coset leader of $C_{\gamma}$ with $\mid C_{\gamma}\mid=4;$
\item[(2)] If $q$ is odd and $\lambda\ge \frac{q+1}{2}$, then $\gamma$ is the coset leader of $C_{\gamma}$ with $\mid C_{\gamma}\mid=4$ except $\mid C_{\frac{q^2+1}{2}}\mid=2$.
\end{itemize}
\end{theorem}

\begin{proof}
The proof is straightforward and follows the same method as Theorem \ref{th6}, so it is omitted.
\end{proof}

\begin{theorem}\label{th6}
Let $m\ge 4$ be even. Then every positive integer $\gamma$ satisfies $1\le \gamma\le \lambda q^{\frac{m}{2}}-1$ and $\gamma\not\equiv 0\pmod{q}$ is a coset leader with $\mid C_{\gamma}\mid =2m$. 
\end{theorem}

\begin{proof}
For any positive integer $1\le \gamma\le \lambda q^{\frac{m}{2}}-1$, assume that $\gamma\equiv a\pmod{\lambda}$ with $1\le a\le \lambda$. Reviewing that $\gamma$ is a coset leader if and only if $\overline{\gamma q^{m-t}}\ge \gamma$ for all $1\le t\le m-1$.

\begin{itemize}
\item[$\mathbf{Case~1:}$]
When $1\le \gamma\le \lambda$, we clearly have
\begin{align*}
\overline{\gamma q^{m-t}}=\gamma q^{m-t}>\gamma
\end{align*}
for any $1\le t\le m-1$.

\item[$\mathbf{Case~2:}$] When $\lambda+1\le \gamma\le \lambda q^{\frac{m}{2}}-1$, we divide the discussion into two cases: $t=1$ and $2\le t\le m-1$.

\begin{itemize}
\item[$\mathbf{Case~2.1:}$] First, we consider the case for $t=1$. Let $\gamma=k\lambda+a$ and $k=Aq+B$, where $0\le B\le q-1$. Since $\lambda+1\le \gamma\le \lambda q^{\frac{m}{2}}-1$, then we have $1\le k\le q^{\frac{m}{2}}-1$ for any $1\le a\le \lambda-1$ and $1\le k\le q^{\frac{m}{2}}-2$ for $a=\lambda$. Further, 
$$Aq+B\le q^{\frac{m}{2}}-1\Leftrightarrow A\le q^{\frac{m}{2}-1}-1~\text{for any }~0\le B\le q-1.$$

When $1\le \gamma q^{m-1}\pmod{n}\le \frac{a}{2}(q^m+1)$ or $a(q^m+1)<\gamma q^{m-1}\pmod{n}\le \frac{\lambda+a}{2}(q^m+1)$, clearly, we have

\begin{align*}
\overline{\gamma q^{m-1}}=&\gamma q^{m-1}\pmod{n}\\
=～&((Aq+B)\lambda+a) q^{m-1}-A \lambda(q^m+1)\\
=～&q^{m-1} (B\lambda+a)-A\lambda>\lambda q^{\frac{m}{2}}>\gamma.
\end{align*}

When $\frac{a}{2}(q^m+1)<\gamma q^{m-1}\pmod{n}<a(q^m+1)$, on one hand, 
\begin{align*}
&\gamma q^{m-1}\pmod{n}>\frac{a}{2} (q^m+1)\\
\Leftrightarrow～& q^{m-1}(B\lambda+a)-A\lambda>\frac{aq}{2}\cdot q^{m-1}+\frac{a}{2}\\
\Leftrightarrow～&B>\frac{a(q-2)}{2\lambda}.
\end{align*}

On the other hand, 
\begin{align*}
&\gamma q^{m-1}\pmod{n}<a(q^m+1)\\
\Leftrightarrow～&q^{m-1}(B\lambda+a)-A\lambda<aq\cdot q^{m-1}+a\\
\Leftrightarrow～&B\le \frac{a(q-1)}{\lambda},
\end{align*}
but when $B=\frac{a(q-1)}{\lambda}$, $\gamma$ is must not be a coset leader because of
\begin{align*}
\gamma=k\lambda+a=(Aq+\frac{a(q-1)}{\lambda})\lambda+a\equiv 0\pmod{q}.
\end{align*}
Therefore, when $\frac{a}{2}(q^m+1)<\gamma q^{m-1}\pmod{n}<a(q^m+1)$, we have
\begin{align*}
\overline{\gamma q^{m-1}}&=a(q^m+1)-\gamma q^{m-1}\pmod{n}\\
&= q^{m-1}(\lambda(a\cdot \frac{q-1}{\lambda}-B))+a+A\lambda\\
&>\lambda q^{\frac{m}{2}}>\gamma.
\end{align*}
 
When $\frac{\lambda+a}{2}(q^m+1)<\gamma q^{m-1}\pmod{n}<n$, we have
\begin{align*}
&\gamma q^{m-1} \pmod{n}>\frac{\lambda+a}{2} (q^m+1)\Leftrightarrow B>\frac{(\lambda+a)q-2a}{2\lambda}
\end{align*}

and

\begin{align*}
\gamma q^{m-1}\pmod{n}<n
\Leftrightarrow B\le q-1.
\end{align*}

Therefore, when $\frac{\lambda+a}{2}(q^m+1)<\gamma q^{m-1}\pmod{n}<n$, we obtain

\begin{align*}
\overline{\gamma q^{m-1}}&=(\lambda+a)(q^m+1)-(\gamma q^{m-1}\pmod{n})\\
&=q^{m-1}((\lambda+a)q-B\lambda-a)+\lambda+a+A\lambda\\
&\ge q^{m-1}(q-\lambda)+\lambda+a+A\lambda\ge \lambda q^{\frac{m}{2}}>\gamma.
\end{align*}

Combining the above arguments, for this subcase, we conclude that $\overline{\gamma q^{m-1}}>\gamma$.

\item[$\mathbf{Case~2.2:}$] Next, we consider the case for $2\le t\le m-1$. Let $\gamma=k\lambda+a$ and $k=Aq^t+B\cdot \frac{q^t-1}{2\lambda}+C$, where $1\le B\le 2\lambda$, $0\le C<\frac{q^t-1}{2\lambda}$. Notice that $m\ge 4$ in this case and reviewing that $\lambda+1\le \gamma\le \lambda q^{\frac{m}{2}}-1$, then we have 
\begin{align*}
A q^t+B\cdot \frac{q^t-1}{2\lambda}+C\le q^{\frac{m}{2}}-1
\Leftrightarrow 0\le A\le\mathrm{min}\{0, q^{\frac{m}{2}-t}-B\cdot \frac{q^t-1}{2\lambda q^t}-\frac{C+1}{q^t}\}.
\end{align*}
When $1\le \gamma q^{m-t}\pmod{n}\le \frac{a}{2}(q^m+1)$ or $a(q^m+1)<\gamma q^{m-t}\pmod{n}\le \frac{\lambda+a}{2}(q^m+1)$, it is clear to have
\begin{align*}
\gamma q^{m-t}\pmod{n}&=q^{m-t}((Aq^t+B\cdot \frac{q^t-1}{2\lambda}+C)\lambda+a)-A\lambda(q^m+1)\\
&=q^{m-t}(B\cdot \frac{q^t-1}{2}+C\lambda+a)-A\lambda>\lambda q^{\frac{m}{2}}>\gamma.
\end{align*}

When $\frac{a}{2}(q^m+1)<\gamma q^{m-t}\pmod{n}<a(q^m+1)$,  on one hand, 
\begin{align*}
&\gamma q^{m-t}\pmod{n}>\frac{a}{2}(q^m+1)\\
\Leftrightarrow~&q^{m-t}(B\cdot \frac{q^t-1}{2}+C\lambda+a)-A\lambda>\frac{aq^t}{2}\cdot q^{m-t}+\frac{a}{2}\\
\Leftrightarrow~&B>a+\frac{2A\lambda+a-aq^{m-t}}{q^{m-t}q^{t}-1}-\frac{2C\lambda}{q^t-1}\\
\Leftrightarrow~&B\ge a.
\end{align*}
In fact, since $0\le C\le \frac{q^t-1}{2\lambda}-1$, $-\frac{a}{q^t-1}<\frac{2A\lambda+a-aq^{m-t}}{q^{m-t}(q^t-1)}<0$, then 
$$-1+\frac{2\lambda-a}{q^t-1}<\frac{2A\lambda+a-aq^{m-t}}{q^{m-t}(q^t-1)}-\frac{2C\lambda}{q^t-1}<0.$$

On the other hand, 
\begin{align*}
&\gamma q^{m-t}\pmod{n}<a(q^m+1)\\
\Leftrightarrow~&q^{m-t}(B\cdot \frac{q^t-1}{2}+C\lambda+a)-A\lambda<aq^t\cdot q^{m-t}+a\\
\Leftrightarrow~&B<2a+\frac{2A\lambda+2a}{q^{m-t}(q^t-1)}-\frac{2C\lambda}{q^t-1}.
\end{align*}

It is easy to find that $B\le 2a$ if $C=0$ and $B\le 2a-1$ otherwise. But when $C=0, B=2a$, $\gamma$ is not a coset leader because of $\gamma=k\lambda+a\equiv 0\pmod{q}.$

Therefore, when $\frac{a}{2}(q^m+1)<\gamma q^{m-t}\pmod{n}<a(q^m+1)$, we have
\begin{align*}
\overline{\gamma q^{m-t}}&=a(q^m+1)-\gamma q^{m-t}\pmod{n}\\
&=q^{m-t}((a-\frac{B}{2})(q^t-1)-C\lambda)+a+A\lambda\\
&>q^{m-t}\cdot \frac{q^t-1}{2}+a+A\lambda\\
&>\lambda q^{\frac{m}{2}}>\gamma.
\end{align*}

When $\frac{\lambda+a}{2}(q^m+1)< \gamma q^{m-t}\pmod{n}<n$, on one hand,
\begin{align*}
&\gamma q^{m-t}\pmod{n}>\frac{\lambda+a}{2}(q^m+1)\\
\Leftrightarrow~&B\ge \lambda+a+\frac{\lambda+a+2A\lambda}{q^{m-t}(q^t-1)}+\frac{\lambda-a-2C\lambda}{q^t-1}.
\end{align*}

We observe that $B\ge \lambda+a+1$ if $C=0$ and $B\ge \lambda+a$ otherwise.

On the other hand, 
\begin{align*}
&\gamma q^{m-t}\pmod{n}<n\\
\Leftrightarrow~& B<2\lambda+\frac{2\lambda(A+1)}{q^{m-t}(q^t-1)}+\frac{2(\lambda-a-C\lambda)}{q^t-1}.
\end{align*} 
It is not difficult to find that $B\le 2\lambda$ if $C=0$ and $B\le 2\lambda-1$ otherwise.

Therefore, when $\frac{\lambda+a}{2}(q^m+1)< \gamma q^{m-t}\pmod{n}<n$, we have 
\begin{align*}
\overline{\gamma q^{m-t}}&=(\lambda+a)(q^m+1)-\gamma q^{m-t}\pmod{n}\\
&=q^{m} ((\lambda-(\frac{B}{2}-a))+q^{m-t}((\frac{B}{2}-C\lambda)-a)+\lambda(A+1)+a\\
&\ge \lambda q^{\frac{m}{2}}>\gamma.
\end{align*}
\end{itemize}
Putting it all together, for this subcase, we conclude that $\overline{\gamma q^{m-t}}>\gamma$.
\end{itemize}
This completes the proof of the theorem.
\end{proof}

Below we turn to the case when $m\ge 3$ is odd. The proof is analogous to that for even $m$ and is omitted here.

\begin{theorem}\label{th13}
Let $m=3$. For any positive integer $\gamma$ with $1\le \gamma\le 2\lambda q-1$ and $\gamma\not\equiv 0\pmod{q}$, we have 
\begin{itemize}
\item[(1)] If $1\le \lambda<\lceil\frac{q}{2}\rceil$, then $\gamma$ is the coset leader of $C_{\gamma}$ with $\mid C_{\gamma}\mid=6$;
\item[(2)] If $\lambda \ge \lceil\frac{q}{2}\rceil$, then $\gamma$ is the coset leader of $C_{\gamma}$ with $\mid C_{\gamma}\mid=6$ except $\mid C_{q^2-q+1}\mid =2$.
\end{itemize}
\end{theorem}

\begin{theorem}\label{th7}
 Let $m\ge 5$ be an odd integer. Then every positive integer $\gamma$ satisfies $1\le \gamma\le \lambda(q^{\frac{m-1}{2}}+q)-1$ and $\gamma\not\equiv 0\pmod{q}$ is a coset leader with $\mid C_{\gamma}\mid =2m$. 
\end{theorem}

In the following part, we give a necessary and sufficient condition for any $0\le \gamma< n$ being a coset leader. For $\lambda=1$, the case was investigated in \cite{zhuhongwei}, here we consider the case with $\lambda> 1$.

For any integer satisfies $0\le \gamma<n$, set $\gamma\equiv a\pmod{\lambda}$ with $1\le a\le \lambda$. We claim that the coset leader of $C_{\gamma}$ is must in the range $[0,\frac{a}{2}(q^m+1)]\cup (a(q^m+1),\frac{\lambda+a}{2}(q^m+1)]$. If the positive integer $\gamma$ satisfies $\frac{a}{2}(q^m+1)<\gamma\le a(q^m+1)$, then we have $\gamma>a(q^m+1)-\gamma\in C_{\gamma}$. Similarly, if positive integer $\gamma$ satisfies $\frac{\lambda+a}{2}(q^m+1)<\gamma<n$, then we have $\gamma>(\lambda+a)(q^m+1)-\gamma\in C_{\gamma}$. Both cases contradict the fact that $\gamma$ is a coset leader.

We begin by defining the following sets.
\begin{itemize}
	\item[(1)] $E_{1} = \{\gamma = aq^{m-t}-A\lambda \ | \ 1 \leq t \leq m-1, 1\le  A<\frac{a(q^{m-t}-1)}{\lambda(q^t+1)}\}$,
	\item[(2)] $E_{2} = \{\gamma = q^{m-t}(B\lambda+a)-A\lambda \ | \ 1 \leq t \leq m-1, A \ \mathrm{and} \ B \ \mathrm{satisfy} \ \mathrm{one} \ \mathrm{of} \ \mathrm{the} \ \mathrm{following} \ \mathrm{conditions}\}$,
	\begin{itemize}
		\item[(2.1)] $1\le  A<\frac{a(q^{m-t}-1)}{2\lambda}$, and $\frac{A\lambda(q^t+1)-a(q^{m-t}-1)}{\lambda(q^{m-t}-1)}<B\le \frac{a(q^m+1)+2A\lambda-2aq^{m-t}}{2\lambda q^{m-t}}$;
		\item[(2.2)] $1\le A\le \frac{a(q^{m-t}-1)}{\lambda}$, and $\frac{a(q^m+1)+A\lambda-a q^{m-t}}{\lambda q^{m-t}}<B\le \frac{(\lambda+a)(q^m+1)+2A\lambda-2a q^{m-t}}{2\lambda q^{m-t}}$;
		\item[(2.3)] $\frac{a(q^{m-t}-1)}{\lambda}<A<\frac{(\lambda+a)(q^{m-t}-1)}{2\lambda}$, and $\frac{A\lambda(q^t+1)-a(q^{m-t}-1)}{\lambda(q^{m-t}-1)}<B\le \frac{(\lambda+a)(q^m+1)+2A\lambda-2a q^{m-t}}{2\lambda q^{m-t}}$,
	\end{itemize}
	\item[(3)] $E_{3} = \{\gamma = (\lambda+a)(q^m+1)-q^{m-t}(B\lambda+a)+A\lambda \ | \ 1 \leq t \leq m-1, A \ \mathrm{and} \ B \ \mathrm{satisfy} \ \mathrm{one} \ \mathrm{of} \ \mathrm{the} \ \mathrm{following} \\ 
	\mathrm{conditions}\}$,
	\begin{itemize}
		\item[(3.1)] $0\le A<\frac{aq^{m-t}-\lambda}{\lambda}$, and $\frac{(\lambda+a)(q^m+1)+2A\lambda-2a\cdot q^{m-t}}{2\lambda q^{m-t}}<B<\frac{A\lambda+\lambda(q^m+1)-aq^{m-t}}{\lambda q^{m-t}}$;
		\item[(3.2)] $\frac{a q^{m-t}-\lambda}{\lambda}\le A<\frac{(\lambda+a)(q^{m-t}-1)}{2\lambda}$, and $
		\frac{(\lambda+a)(q^m+1)+2A\lambda-2a\cdot q^{m-t}}{2\lambda q^{m-t}}<B<\frac{(\lambda+a)(q^m+1)-a(q^{m-t}+1)-A\lambda(q^t-1)}{\lambda(q^{m-t}+1)}$.
	\end{itemize}
\end{itemize}
Finally, we set $E = E_{1} \cup E_{2} \cup E_{3}$.
	
\begin{theorem}\label{th4}
For any $0 \leq \gamma < n$, let $a$ be the integer in the range $1 \leq a \leq \lambda$ such that $\gamma \equiv a \pmod{\lambda}$. Then $\gamma$ is a coset leader if and only if both of the following conditions hold:
	\begin{itemize}
		\item[(1)] $0 \leq \gamma \leq \frac{a}{2}(q^m+1)$ or $a(q^m+1)< \gamma\le \frac{\lambda+a}{2}(q^m+1)$;
		\item[(2)] $\gamma \notin E$.
	\end{itemize}
\end{theorem}

\begin{proof}
For any positive integer $\gamma$ satisfying $0\le \gamma<n$, suppose $\gamma\equiv a\pmod{\lambda}$ such that $1\le a\le \lambda$. For any integer $t$ with $1\le t\le m-1$, we express $\gamma$ in the form $\gamma=(Aq^t+B)\lambda+a$, where $A,B$ are positive integers satisfies $0\le B\le q^t-1$ for $0\le A\le q^{m-t}-1$ and $B=0$ for $A=q^{m-t}$. We split the proof into two cases.

\begin{itemize}
\item[$\mathbf{Case~1: }$] Let $B=0$. We  divide this case into the following subcases for discussion.
\begin{itemize}
\item[$\mathbf{Case~1.1: }$] When $0\le A\le \frac{aq^{m-t}}{\lambda}$. We derive 
\begin{align*}
\gamma q^{m-t}\pmod{n} =aq^{m-t}-A\lambda\le \frac{a}{2}(q^m+1).
\end{align*}
It follows that
\begin{align*}
\overline{\gamma q^{m-t}}=\gamma q^{m-t}\pmod{n}=aq^{m-t}-A\lambda.
\end{align*}
In addition, from the equivalence relation
\begin{align*}
aq^{m-t}-A\lambda>A\lambda q^t+a\Leftrightarrow A<\frac{a(q^{m-t}-1)}{\lambda(q^t+1)},
\end{align*}
we can conclude that for any integers $1\le a\le \lambda$, $1\le t\le m-1$ and $0\le A<\frac{a(q^{m-t}-1)}{\lambda(q^t+1)}$, the integers $aq^{m-t}-A\lambda$ are not coset leaders.

\item[$\mathbf{Case~1.2: }$] When $\frac{aq^{m-t}}{\lambda}<A\le q^{m-t}$. We derive the congruence relation
\begin{align*}
\gamma q^{m-t}\pmod{n}
=aq^{m-t}-A\lambda+\lambda(q^m+1)>\frac{\lambda+a}{2}(q^m+1).
\end{align*}
Accordingly, we obtain 
\begin{align*}
\overline{\gamma q^{m-t}}&=(\lambda+a)(q^m+1)-(aq^{m-t}-A\lambda+\lambda(q^m+1))\\
&=a(q^m+1)-aq^{m-t}+A\lambda.
\end{align*}
From the equivalence relation
$$a(q^m+1)-aq^{m-t}+A\lambda>A\lambda q^{t}+a\Leftrightarrow A<\frac{aq^{m-t}}{\lambda},$$
we note that this result contradicts the range condition $\frac{aq^{m-t}}{\lambda}<A$ of the current subcase. Thus, for any integers $1\le a\le \lambda-1$, $1\le t\le m-1$ and $\frac{aq^{m-t}}{\lambda}<A\le \frac{(\lambda+a)q^{m}+(\lambda-a)}{2\lambda q^t}$, we cannot draw the conclusion that $a(q^m+1)-aq^{m-t}+A\lambda$ are not coset leaders.
\end{itemize}

\item[$\mathbf{Case~2: }$] Let $B\ne 0$. For any integer $t$ satisfying $1\le t\le m-1$, we derive
\begin{align*}
\gamma q^{m-t}\pmod{n}
=q^{m-t}(B\lambda+a)-A\lambda.
\end{align*}
We further decompose this case into the following four subcases for detailed analysis.

\begin{itemize}
\item[$\mathbf{Case~2.1:}$] Suppose that $0\le\gamma q^{m-t}\pmod{n}\le\frac{a}{2}(q^m+1)$. This condition is equivalent to the inequality
\begin{align*}
1\le B\le \frac{a(q^m+1)+2A\lambda-2aq^{m-t}}{2\lambda q^{m-t}}.
\end{align*}
In this scenario, we have
\begin{align*}
\overline{\gamma q^{m-t}}=q^{m-t}(B\lambda+a)-A\lambda.
\end{align*}
Moreover, the inequality $\overline{\gamma q^{m-t}}>\gamma$ holds if and only if 
\begin{align*}
 B>\frac{A\lambda(q^t+1)-a(q^{m-t}-1)}{\lambda(q^{m-t}-1)}.
\end{align*}
Assume that the upper bound of $B$ is greater than its lower bound derived above, i.e.,
\begin{align*}
\frac{a(q^m+1)+2A\lambda-2aq^{m-t}}{2\lambda q^{m-t}}>\frac{A\lambda(q^t+1)-a(q^{m-t}-1)}{\lambda(q^{m-t}-1)}.
\end{align*}
This assumption can be simplified to the constraint on $A$:
\begin{align*}
A<\frac{a(q^{m-t}-1)}{2\lambda}.
\end{align*}
Furthermore, we get
\begin{align*}
\frac{a(q^m+1)+2A\lambda-2a q^{m-t}}{2\lambda q^{m-t}}<\frac{a(q^t-1)}{2\lambda}<q^t.
\end{align*}
Accordingly, for any integers $A$ and $B$ satisfying $0\le  A<\frac{a(q^{m-t}-1)}{2\lambda}$ and $\frac{A\lambda(q^t+1)-a(q^{m-t}-1)}{\lambda(q^{m-t}-1)}<B\le \frac{a(q^m+1)+2A\lambda-2aq^{m-t}}{2\lambda q^{m-t}}$, the integers $q^{m-t}(B\lambda+a)-A\lambda$ are not coset leaders.

\item[$\mathbf{Case~2.2: }$] Suppose that $\frac{a}{2}(q^m+1)< \gamma q^{m-t}\pmod{n}< a(q^m+1)$. This condition is equvalent to the following inequality for $B$:
\begin{align*}
\frac{a(q^m+1)+2A\lambda-2aq^{m-t}}{2\lambda q^{m-t}}< B<\frac{a(q^m+1)+A\lambda-a q^{m-t}}{\lambda q^{m-t}}.
\end{align*}
And in this subcase, we have
$$\overline{\gamma q^{m-t}}=a(q^m+1)-q^{m-t}(B\lambda+a)+A\lambda.$$

Moreover, the inequality $\overline{\gamma q^{m-t}}>\gamma$ holds if and only if
$$B<\frac{(q^t-1)(aq^{m-t}-A\lambda)}{\lambda(q^{m-t}+1)}.$$

Notice that
$$\frac{(q^t-1)(aq^{m-t}-A\lambda)}{\lambda(q^{m-t}+1)}<\frac{a(q^m+1)+A\lambda-a q^{m-t}}{\lambda q^{m-t}}$$
holds for all positive integer $A$, and 
$$\frac{(q^t-1)(aq^{m-t}-A\lambda)}{\lambda(q^{m-t}+1)}>\frac{a(q^m+1)+2A\lambda-2aq^{m-t}}{2\lambda q^{m-t}}\Leftrightarrow A<\frac{a(q^{m-t}-1)}{2\lambda}.$$

Additionally, 
\begin{align*}
\frac{(q^t-1)(a q^{m-t}-A\lambda)}{\lambda(q^{m-t}+1)}&\le\frac{a (q^t-1)}{\lambda} \frac{q^{m-t}}{q^{m-t}+1}< q^t-1.
\end{align*}

Therefore, for any integers $A, B$ satisfying $0\le A<\frac{a(q^{m-t}-1)}{2\lambda}$ and $\frac{a(q^m+1)+2A\lambda-2a q^{m-t}}{2\lambda q^{m-t}}<B<\frac{(q^t-1)(a q^{m-t}-A\lambda)}{\lambda(q^{m-t}+1)}$, the integers $a(q^m+1)-q^{m-t}(B\lambda+a)+A\lambda$ are not coset leaders.

\item[$\mathbf{Case~2.3: }$] Suppose that $a(q^m+1)<\gamma q^{m-t}\pmod{n}\le  \frac{\lambda+a}{2}(q^m+1)$($1\le a\le \lambda-1$). This condition is equvalent to the following inequality for $B$:
\begin{align*}
\frac{a(q^m+1)+A\lambda-a q^{m-t}}{\lambda q^{m-t}}<B\le \frac{(\lambda+a)(q^m+1)+2A\lambda-2a q^{m-t}}{2\lambda q^{m-t}}.
\end{align*}
And in this subcase, we have
\begin{align*}
\overline{\gamma q^{m-t}}=q^{m-t}(B\lambda+a)-A\lambda.
\end{align*}

We know from Case 2.1 that
\begin{align*}
\overline{\gamma q^{m-t}}>\gamma
\Leftrightarrow B>\frac{A\lambda(q^t+1)-a(q^{m-t}-1)}{\lambda(q^{m-t}-1)}.
\end{align*}
Let
\begin{align*}
\frac{A\lambda(q^t+1)-a(q^{m-t}-1)}{\lambda(q^{m-t}-1)}\le \frac{a(q^m+1)+A\lambda-a q^{m-t}}{\lambda q^{m-t}},
\end{align*}
we have
\begin{align*}
A\le \frac{a (q^{m-t}-1)}{\lambda}.
\end{align*}
We know
\begin{align*}
\frac{a(q^m+1)+A\lambda-a q^{m-t}}{\lambda q^{m-t}}<\frac{(\lambda+a)(q^m+1)+2A\lambda-2a q^{m-t}}{2\lambda q^{m-t}}
\end{align*}
at any time. Then we consider the magnitude relationship between $\frac{(\lambda+a)(q^m+1)+2A\lambda-2a q^{m-t}}{2\lambda q^{m-t}}$ and $\frac{A\lambda(q^t+1)-a(q^{m-t}-1)}{\lambda(q^{m-t}-1)}$. Clearly, we have 
\begin{align*}
\frac{A\lambda(q^t+1)-a(q^{m-t}-1)}{\lambda(q^{m-t}-1)}<\frac{(\lambda+a)(q^m+1)+2A\lambda-2a q^{m-t}}{2\lambda q^{m-t}}
\Leftrightarrow A<\frac{(\lambda+a)(q^{m-t}-1)}{2\lambda},
\end{align*}
and  
\begin{align*}
\frac{a(q^{m-t}-1)}{\lambda}<\frac{(\lambda+a)(q^{m-t}-1)}{2\lambda}.
\end{align*}

Therefore, for any $1\le a\le \lambda-1$, and any positive integers satisfying  $0\le A\le \frac{a(q^{m-t}-1)}{\lambda}$, $\frac{a(q^m+1)+A\lambda-a q^{m-t}}{\lambda q^{m-t}}<B\le \frac{(\lambda+a)(q^m+1)+2A\lambda-2a q^{m-t}}{2\lambda q^{m-t}}$ or $\frac{a(q^{m-t}-1)}{\lambda}<A<\frac{(\lambda+a)(q^{m-t}-1)}{2\lambda}$, $\frac{A\lambda(q^t+1)-a(q^{m-t}-1)}{\lambda(q^{m-t}-1)}<B\le \frac{(\lambda+a)(q^m+1)+2A\lambda-2a q^{m-t}}{2\lambda q^{m-t}}$, the integers $q^{m-t}(B\lambda+a)-A\lambda$ are not coset leaders.

\item[$\mathbf{Case~2.4: }$] Suppose that $\frac{\lambda+a}{2}(q^m+1)<\gamma q^{m-t}< \lambda(q^m+1), (1\le a\le \lambda-1)$. This condition is equvalent to the following inequality for $B$:
$$\frac{(\lambda+a)(q^m+1)+2A\lambda-2a q^{m-t}}{2\lambda q^{m-t}}<B<\frac{A\lambda+\lambda(q^m+1)-aq^{m-t}}{\lambda q^{m-t}}.$$
And in this subcase, we have 

\begin{align*}
\overline{\gamma q^{m-t}}=(\lambda+a)(q^m+1)-q^{m-t}(B\lambda+a)+A\lambda.
\end{align*}
Moreover, the inequality $\overline{\gamma q^{m-t}}>\gamma$ holds if and only if
$$B<\frac{(\lambda+a)(q^m+1)-a(q^{m-t}+1)-A\lambda(q^t-1)}{\lambda(q^{m-t}+1)}.
$$
If we let 
\begin{align*}
\frac{(\lambda+a)(q^m+1)-a(q^{m-t}+1)-A\lambda(q^t-1)}{\lambda(q^{m-t}+1)}\le \frac{A\lambda+\lambda(q^m+1)-aq^{m-t}}{\lambda q^{m-t}},
\end{align*}
we have 
\begin{align*}
A\ge \frac{a q^{m-t}-\lambda}{\lambda}.
\end{align*}
We know
\begin{align*}
\frac{(\lambda+a)(q^m+1)+2A\lambda-2a q^{m-t}}{2\lambda q^{m-t}}<\frac{A\lambda+\lambda(q^m+1)-aq^{m-t}}{\lambda q^{m-t}}
\end{align*}
at any time, then we consider the magnitude relationship between $\frac{(\lambda+a)(q^m+1)-a(q^{m-t}+1)-A\lambda(q^t-1)}{\lambda(q^{m-t}+1)}$ and $\frac{(\lambda+a)(q^m+1)+2A\lambda-2a q^{m-t}}{2\lambda q^{m-t}}$. Obviously, we have 
\begin{align*}
&\frac{(\lambda+a)(q^m+1)-a(q^{m-t}+1)-A\lambda(q^t-1)}{\lambda(q^{m-t}+1)}>\frac{(\lambda+a)(q^m+1)+2A\lambda-2a q^{m-t}}{2\lambda q^{m-t}}\\
\Leftrightarrow& A<\frac{(\lambda+a)(q^{m-t}-1)}{2\lambda}\le q^{m-t}-1
\end{align*}
and
\begin{align*}
\frac{aq^{m-t}-\lambda}{\lambda}<\frac{(\lambda+a)(q^{m-t}-1)}{2\lambda}.
\end{align*}
Therefore, for any $1\le a\le \lambda-1$, and any positive integers satisfying $0\le A<\frac{aq^{m-t}-\lambda}{\lambda}$, $\frac{(\lambda+a)(q^m+1)+2A\lambda-2a\cdot q^{m-t}}{2\lambda q^{m-t}}<B<\frac{A\lambda+\lambda(q^m+1)-aq^{m-t}}{\lambda q^{m-t}}$ or $\frac{a q^{m-t}-\lambda}{\lambda}\le A<\frac{(\lambda+a)(q^{m-t}-1)}{2\lambda}$, $
\frac{(\lambda+a)(q^m+1)+2A\lambda-2a\cdot q^{m-t}}{2\lambda q^{m-t}}<B<\frac{(\lambda+a)(q^m+1)-a(q^{m-t}+1)-A\lambda(q^t-1)}{\lambda(q^{m-t}+1)}$, the integers $(\lambda+a)(q^m+1)-q^{m-t}(B\lambda+a)+A\lambda$ are not coset leaders.
\end{itemize}
\end{itemize}
Summarizing the discussion above, the desired theorem then follows.
\end{proof}

\begin{example}
Let $q=5$ and $m=2$, let $0\le \gamma< 103$, from Theorem \ref{th4}, we have $\gamma$ is a coset leader if and only if 
\begin{align*}
\gamma\in \{0,1,2,3,4,6,7,8,9,13,14,16,26,27,29,39\}.
\end{align*}
The result is verified by Magma program.
\end{example}

For $\lambda=1$, partial results regarding the largest coset leaders can be found in \cite{liuy, yanhaode, zhuhongwei}. In the following, let $\delta_1$ and $\delta_2$ denote the largest and second largest $q$-cyclotomic coset leaders modulo $n$, respectively. We proceed to determine the values of them for $\lambda\mid q-1$.

\begin{lemma}\label{lem5}
Let $n=\lambda(q^m+1)$, where $\lambda\mid q-1$. Then the largest coset leader of $q$-cyclotomic coset modulo $n$ satisfies $\delta_1\equiv \lambda-1\pmod{\lambda}$. Specially, $(\lambda-1)(q^m+1)+1$ is a coset leader with $\mid C_{(\lambda-1)(q^m+1)+1}\mid =2m$.
\end{lemma}

\begin{proof}

As the conclusion holds trivially when $\lambda=1$, we focus on the case $\lambda>1$ in the following.

For any positive integer $0\le \gamma<n$, set $\gamma\equiv a\pmod{\lambda}$, where $1\le a\le \lambda$. If $\gamma$ is a coset leader and $a\ne \lambda-1$, then $\gamma$ satisfies $0\le \gamma\le (\lambda-1)(q^m+1)$. However, we know that there may exist coset leaders $\gamma$ with $\gamma\equiv \lambda-1\pmod{\lambda}$ such that $(\lambda-1)(q^m+1)<\gamma\le \frac{2\lambda-1}{2}(q^m+1)$. Therefore, if we find one coset leader $\gamma$ satisfying $(\lambda-1)(q^m+1)<\gamma\le \frac{2\lambda-1}{2}(q^m+1)$, we can conclude that $\delta_1$ satisfies $\delta_1\equiv \lambda-1\pmod{\lambda}$. Below, we will show that $(\lambda-1)(q^m+1)+1$ is a coset leader. 

For any integer $t$ with $1\le t\le m-1$, we have 
\begin{align*}
&((\lambda-1)(q^m+1)+1)q^t\\
=~&(\frac{q^t-1}{\lambda}\cdot \lambda+1)\cdot (\lambda-1)(q^m+1)+q^t\\
=~&(q^t-\frac{q^t-1}{\lambda})\cdot \lambda(q^m+1)-(q^m+1)+q^t\\
\equiv ~&(\lambda-1)(q^m+1)+q^t\pmod{n}.
\end{align*}
And for all $1\le t\le m-1$, it follows that 
$$(\lambda-1)(q^m+1)+1<(\lambda-1)(q^m+1)+q^t<\frac{2\lambda-1}{2}(q^m+1),$$
this implies that
$$\overline{((\lambda-1)(q^m+1)+1)q^t}>(\lambda-1)(q^m+1)+1,$$
and hence $(\lambda-1)(q^m+1)+1$ is a coset leader. Moreover, since 
$$\mathrm{ord}_{\frac{\lambda(q^m+1)}{\mathrm{gcd}(\lambda(q^m+1),(\lambda-1)(q^m+1)+1)}}(q)=\mathrm{ord}_{\lambda(q^m+1)}(q)=2m,$$
we conclude that $\mid C_{(\lambda-1)(q^m+1)+1}\mid =2m$. 

Therefore, we finish the proof of this lemme.
\end{proof}

\begin{theorem}\label{th9}
Let $n=\lambda(q^m+1)$, where $m\ge 1$ is an odd positive integer. Then 
\begin{equation*}
\delta_1=	\left\{
	\begin{array}{lcl}
	\frac{q^m+1}{q+1}(\frac{2\lambda q+\lambda-q-1}{2}), \ \mathrm{if}~ \frac{q-1}{\lambda}~\mathrm{is}~odd;\\
\frac{2\lambda-1}{2}(q^m+1), \ \mathrm{if}~\frac{q-1}{\lambda}~\mathrm{is}~even.\\
	\end{array} \right.
\end{equation*}
and 
\begin{equation*}
	\mid C_{\delta_1}\mid=	\left\{
	\begin{array}{lcl}
	2, \ \mathrm{if}~ \frac{q-1}{\lambda}~\mathrm{is~odd};\\
		1, \ \mathrm{if}~\frac{q-1}{\lambda}~\mathrm{is}~even.
	\end{array} \right.
\end{equation*}
\end{theorem}

\begin{proof}
According to Lemma \ref{lem5}, we have the largest coset leader $\delta_1$ satisfies $\delta_1\equiv \lambda-1\pmod{\lambda}$, and reviewing 
that $(\lambda-1)(q^m+1)<\delta_1\le \frac{2\lambda-1}{2}(q^m+1)$. We now prove the result by considering the following two cases.

When $\frac{q-1}{\lambda}$ is even, since 
$$C_{\frac{2\lambda-1}{2}(q^m+1)}=\left\{\frac{2\lambda-1}{2}(q^m+1)\right\},$$
then we can conclude $\delta_1=\frac{2\lambda-1}{2}(q^m+1)$ directly. \\

When $\frac{q-1}{\lambda}$ is odd, since
$$C_{\frac{q^m+1}{q+1}(\frac{2\lambda q+\lambda-q-1}{2})}=\left\{\frac{q^m+1}{q+1}(\frac{2\lambda q+\lambda-q-1}{2} ),\frac{q^m+1}{q+1}(\frac{2\lambda q+3\lambda-q-1}{2})\right\},$$
then it is obvious to have $\frac{q^m+1}{q+1}(\frac{2\lambda q+\lambda-q-1}{2})$ is a coset leader. Below, we need to show that any integer $\gamma$ satisfying  
$$\frac{q^m+1}{q+1}(\frac{2\lambda q+\lambda-q-1}{2})<\gamma\le \frac{2\lambda-1}{2}(q^m+1)$$
and $\gamma\equiv \lambda-1\pmod{\lambda}$ cannot be a coset leader. Let $\gamma=\frac{q^m+1}{q+1}(\frac{2\lambda q+\lambda-q-1}{2})+k\lambda$, $\gamma^{\prime}=(2\lambda-1)(q^m+1)-\gamma$, where $k$ is a positive integer. Recall that $\gamma$ is not a coset leader if and only if there exists a positive integer $t$ such that $\overline{\gamma q^{t}}<\gamma$. We will proof that $\overline{\gamma q}<\gamma$. Considering the congruence
$$\frac{q^m+1}{q+1}(\frac{2\lambda q+3\lambda-q-1}{2})\equiv \frac{q^m+1}{q+1}(\frac{2\lambda q+\lambda-q-1}{2})q\pmod{n}.$$
This can be written as
$$\frac{q^m+1}{q+1}(\frac{2\lambda q+\lambda-q-1}{2})q=\frac{q^m+1}{q+1}(\frac{2\lambda q+3\lambda-q-1}{2})-an$$
for some positive integer $a$. Thus, showing $\overline{\gamma q}<\gamma$ is equvalent to show
$$(2\lambda-1)(q^m+1)-\gamma<\gamma q-an<n+\gamma.$$
If we assume
$$\frac{q^m+1}{q+1}(\frac{2\lambda q+3\lambda-q-1}{2})<\gamma q-an<n+\frac{q^m+1}{q+1}(\frac{2\lambda q+\lambda-q-1}{2}),$$
we obtain
$$1\le k<\frac{n(a+1)-\delta_1(q-1)}{\lambda q}=\frac{q^m+1}{q+1}.$$
This implies that for any 
$$\frac{q^m+1}{q+1}(\frac{2\lambda q+\lambda-q-1}{2})+\lambda\le \gamma\le \frac{2\lambda-1}{2}(q^m+1)<\frac{q^m+1}{q+1}(\frac{2\lambda q+3\lambda-q-1}{2}),$$
we can always find $\overline{\gamma q}<\gamma$. Therefore, no such $\gamma$ can be a coset leader. It follows that $\delta_1=\frac{q^m+1}{q+1}(\frac{2\lambda q+\lambda-q-1}{2})$.

Combining both cases, the proof of the theorem is complete.
\end{proof}

\begin{theorem}\label{th14}
Let $n=\lambda(q^m+1)$, where $m\ge 3$ is an odd positive integer. Then 
\begin{equation*}
\delta_2=	\left\{
	\begin{array}{lcl}
\delta_1-\lambda q\frac{q^{m-2}+1}{q+1}, \ \mathrm{if}~ \frac{q-1}{\lambda}~\mathrm{is}~odd;\\
\delta_1-\frac{\lambda(q^m+1)}{q+1}, \ \mathrm{if}~\frac{q-1}{\lambda}~\mathrm{is}~even.\\
	\end{array} \right.
\end{equation*}
and 
\begin{equation*}
\mid C_{\delta_2}\mid=	\left\{
	\begin{array}{lcl}
2m, \ \mathrm{if}~ \frac{q-1}{\lambda}~\mathrm{is}~odd;\\
2,, \ \mathrm{if}~\frac{q-1}{\lambda}~\mathrm{is}~even.\\
	\end{array} \right.
\end{equation*}
\end{theorem}

\begin{proof}
We divide the proof into two cases: $\frac{q-1}{\lambda}$ is odd and $\frac{q-1}{\lambda}$ is even.

\begin{itemize}
\item[$\mathbf{Case~1:}$] When $\frac{q-1}{\lambda}$ is even. Let $\gamma=\frac{q^m+1}{q+1}(\frac{2\lambda q-q-1}{2})$, the coset is given by
$$C_{\gamma}=\left\{\frac{q^m+1}{q+1}(\frac{2\lambda q-q-1}{2}),\frac{q^m+1}{q+1}(\frac{2\lambda q+4\lambda-q-1}{2})\right\}.$$
It is obvious that $\gamma=\frac{q^m+1}{q+1}(\frac{2\lambda q-q-1}{2})$ is a coset leader. We next prove thet for any $\gamma^{\prime}$ satisfying  $\gamma<\gamma^{\prime}<\delta_1$ and $\gamma^{\prime}\equiv \lambda-1\pmod{\lambda}$, $\gamma^{\prime}$ is not a coset leader. Let $\gamma^{\prime}=\gamma+k\lambda$, where $1\le k<\frac{q^m+1}{q+1}$, consider the element $\overline{\gamma}^{\prime}=\gamma q-an$, where $a$ is a positive integer. We now show that for any $1\le k<\frac{q^m+1}{q+1}$, we all have
$$\overline{\gamma}^{\prime}<\gamma^{\prime}q-an<n+\gamma^{\prime}.$$
On one hand, 
$$\overline{\gamma}^{\prime}<\gamma^{\prime}q-an\Leftrightarrow k\ge 1.$$
On the other hand, 
\begin{align*}
&\gamma^{\prime} q-an<n+\gamma^{\prime}\\
\Leftrightarrow~&(\gamma+k\lambda)q-an<n+\gamma=k\lambda\\
\Leftrightarrow~&k<\frac{q^m+1}{q+1},
\end{align*}
which means that for any $\gamma<\gamma^{\prime}<\delta_1$ and $\gamma^{\prime}\equiv \lambda-1\pmod{\lambda}$, we all can find $\overline{\gamma q}<\gamma$. Then we conclude that $\delta_2=\frac{q^m+1}{q+1}(\frac{2\lambda q-q-1}{2})$.

\item[$\mathbf{Case~2:}$] When $\frac{q-1}{\lambda}$ is odd. Let $\gamma=\delta_1-\lambda q\frac{q^{m-2}+1}{q+1}$, we first proof that $\gamma$ is a coset leader, which is equvalent to proof for any $1\le t\le m-1$, we all have $\overline{\gamma q^t}>\gamma$, that is $\gamma<\gamma q^t\pmod{n}<(2\lambda-1)(q^m+1)-\gamma$. We consider the cases $t$ odd and $t$ even separately.

For odd $t$. We have 
\begin{align*}
\gamma q^t\pmod{n}&\equiv (\delta_1-\lambda q\frac{q^{m-2}+1}{q+1})\cdot q^t\pmod{n}\\
&\equiv \overline{\delta_1}^{\prime}-\lambda(q^{m+t-2}-q^{m+t-3}+\cdots +q^{m+1}-q^m)-\lambda(q^{m-1}-q^{m-2}+\cdots +q^{t+1})\\
&\equiv \overline{\delta_1}^{\prime}-\lambda q^m(q^{t-2}-q^{t-3}+\cdots +q-1)-\lambda q^{t+1}\frac{q^{m-t-1}+1}{q+1}\\
&\equiv \overline{\delta_1}^{\prime}+\lambda\frac{q^{t-1}-1}{q+1}-\lambda q^{t+1}\frac{q^{m-t-1}+1}{q+1}\\
&\equiv \delta_1-\lambda q^{t-1}(q-1)\pmod{n}.
\end{align*}
It is straightforward to see that
$$\gamma q^t\pmod{n}=\delta_1-\lambda q^{t-1}(q-1)\ge \delta_1-\lambda q^{m-3}(q-1)>\gamma.$$
Next, we show that $\gamma q^t\pmod{n}<(2\lambda-1)(q^m+1)-\gamma$. We compute
\begin{align*}
&\gamma q^t\pmod{n}-((2\lambda-1)(q^m+1)-\gamma)\\
=& \frac{q^m+1}{q+1}(2\lambda q+\lambda-q-1)-\lambda q^{t-1}(q-1)-(2\lambda-1)(q^m+1)-\lambda q\frac{q^{m-2}+1}{q+1}\\
=&-\lambda\frac{q^m+1+q^{t-1}(q^2-1)+q(q^{m-2}+1)}{q+1}<0.
\end{align*}
Therefore, for any odd $t$ with $1\le t\le m-1$, we have $\overline{\gamma q^t}>\gamma$.

For even $t$. By the same way, we have
\begin{align*}
\gamma q^t\pmod{n}\equiv &(\delta_1-\lambda q\frac{q^{m-2}}{q+1})q^t\\
\equiv &\delta_1+\lambda\frac{q^{t-1}+1}{q+1}+\lambda q^{t+1}\frac{q^{m-t-1}-1}{q+1}\pmod{n}.
\end{align*}
On one hand, It is obvious to have $\gamma q^t\pmod{n}>\gamma$. On the other hand, 
\begin{align*}
&\gamma q^t\pmod{n}-((2\lambda-1)(q^m+1)-\gamma)\\
 =~&\delta_1+\lambda\frac{q^{t-1}+1}{q+1}+\lambda q^{t+1}\frac{q^{m-t-1}-1}{q+1}-(2\lambda-1)(q^m+1)+\delta_1-\lambda q\frac{q^{m-2}+1}{q+1}\\
=~&-\lambda\frac{q^{t-1}(q^2-1)+q(q^{m-2}+1)}{q+1}<0.
\end{align*}
Therefore, for any even $t$ with $1\le t\le m-1$, we have $\overline{\gamma q^t}>\gamma$. This shows that $\gamma$ is a coset leader.

We now show that $\gamma$ is the second largest coset leader. This is equvalent to proving that for any $\gamma<\gamma^{\prime}<\delta_1$ and $\gamma^{\prime}\equiv \lambda-1\pmod{\lambda}$, there exists at least one $t$ with $1\le t\le m-1$ such that $\overline{\gamma^{\prime}q^t}<\gamma^{\prime}$. In other words, we need to show that $\gamma^{\prime} q^t\pmod{n}>(2\lambda-1)(q^m+1)-\gamma^{\prime}$ or $\gamma^{\prime} q^t\pmod{n}<\gamma^{\prime}$. Let $\gamma^{\prime}=\gamma+k\lambda$, where $1\le k< q\frac{q^{m-2}+1}{q+1}$. Notice that 
$$q\frac{q^{m-2}+1}{q+1}=(q-1)(q^{m-3}+q^{m-5}+\cdots +q^2+1)+1.$$
We then split the proof into two parts: $1\le k\le q-1$ and $q-1<k<q\frac{q^{m-2}+1}{q+1}$.

For $1\le k\le q-1$. We consider $\gamma^{\prime}q^{m-1}\pmod{n}$, since
\begin{align*}
\gamma q^{m-1}&= (\delta_1-\lambda q\frac{q^{m-2}+1}{q+1})q^{m-1}\\
&\equiv \delta_1-\lambda\frac{q^{m-2}+1}{q+1}q^m\\
&\equiv \delta_1-\frac{q^{m-2}+1}{q+1}\cdot \lambda(q^m+1)+\lambda\frac{q^{m-2}+1}{q+1}\\
&\equiv \delta_1+\lambda\frac{q^{m-2}+1}{q+1}\pmod{n},
\end{align*}
it follows that 
$$\gamma^{\prime} q^{m-1}\equiv \delta_1+\lambda\frac{q^{m-2}+1}{q+1}+k\lambda q^{m-1}\pmod{n}.$$
We proceed to show that
$$(2\lambda-1)(q^m+1)-\gamma^{\prime}<\delta_1+\lambda\frac{q^{m-2}+1}{q+1}+k\lambda q^{m-1}<n+\gamma^{\prime}.$$
On one hand, 
\begin{align*}
&(2\lambda-1)(q^m+1)-\gamma^{\prime}<\delta_1+\lambda\frac{q^{m-2}+1}{q+1}+k\lambda q^{m-1}\\
\Leftrightarrow ~&k\lambda (q^m+1)>(2\lambda-1)(q^m+1)-\gamma+\delta_1-\lambda\frac{q^{m-2}+1}{q+1}.
\end{align*}
Notice that this inequality holds for $k=1$, as
$$\lambda(q^{m-1}+1)>(2\lambda-1)(q^m+1)-\gamma+\delta_1-\lambda\frac{q^{m-2}+1}{q+1},$$
Thus, the inequality is satisfied for all $1\le k\le q-1$. On the other hand, 
\begin{align*}
&\delta_1+\lambda\frac{q^{m-2}+1}{q+1}+k\lambda q^{m-1}<n+\gamma^{\prime}\\
\Leftrightarrow ~&k\lambda(q^{m-1}-1)<n+\gamma-\delta_1-\lambda\frac{q^{m-2}+1}{q+1}\\
\Leftrightarrow~& k<\frac{q^{m-2}(q^2-1)}{q^{m-1}-1}=q-\frac{q^{m-3}-1}{q^{m-1}-1}\\
\Leftrightarrow~& k\le q-1.
\end{align*}
Therefore, for any $1\le k\le q-1$, we can take $t=m-1$ such that $\overline{\gamma^{\prime} q^t}<\gamma^{\prime}$. This implies that $\gamma^{\prime}$ is not a coset leader.

For $q-1<k<q\frac{q^{m-2}+1}{q+1}$. For such $k$, we can always find an integer $h$ with $0\le h\le m-5$ such that 
$$(q-1)(q+q^2+\cdots +q^h)<k\le (q-1)(1+q^2+\cdots +q^{h+2}).$$
Let $k^{\prime}=k-(q-1)(q+q^2+\cdots q^h)$. We consider $\gamma^{\prime} q^{m-(h+3)}\pmod{n}$. First, we write $\gamma^{\prime}$
\begin{align*}
\gamma^{\prime}&=\gamma+k\lambda\\
&=\delta_1-\lambda q\frac{q^{m-2}+1}{q+1}+\lambda((q-1)(1+q^2+\cdots +q^h))+k^{\prime}\lambda\\
&=\delta_1-\lambda q^{h+3}\frac{q^{m-h-4}+1}{q+1}+\lambda(q^{h+2}-1)+k^{\prime}\lambda,
\end{align*}
then 
\begin{align*}
\gamma^{\prime} q^{m-(h+3)}\equiv &(\delta_1-\lambda q^{h+3}\frac{q^{m-h-4}+1}{q+1}+\lambda(q^{h+2}-1)+k^{\prime}\lambda)q^{m-(h+3)}\\
\equiv &\delta_1+\lambda\frac{q^{m-t-4}+1}{q+1}+\lambda(q^{t+2}-1)q^{m-(t+3)}+k^{\prime}\lambda q^{m-(t+3)}\pmod{n}.
\end{align*}
We proceed to show that this element lies in the desired range:
$$(2\lambda-1)(q^m+1)-\gamma^{\prime}<\delta_1+\lambda\frac{q^{m-h-4}+1}{q+1}+\lambda(q^{h+2}-1)q^{m-(h+3)}+k^{\prime}\lambda q^{m-(h+3)}<n+\gamma^{\prime}.$$
On one hand, consider the left inequality:
\begin{align*}
&\delta_1+\lambda\frac{q^{m-h-4}+1}{q+1}+\lambda(q^{h+2}-1)q^{m-(h+3)}+k^{\prime}\lambda q^{m-(h+3)}>(2\lambda-1)(q^m+1)-\gamma^{\prime}\\
\Leftrightarrow &k^{\prime}(q+1)(q^{m-(h+3)}+1)>q^{m-h-2}+q^{m-h-3}-q^{m-h-4}-q^{h+2}+q+1.
\end{align*}
Notice that when $k^{\prime}=1$, we have 
\begin{align*}
(q+1)(q^{m-(h+3)}+1)&=q^{m-h-2}+q^{m-h-3}+q+1\\
&>q^{m-h-2}+q^{m-h-3}+q+1-(q^{m-h-4}+q^{h+2}), 
\end{align*}
then for any $1\le k^{\prime}\le (q-1)q^{h+2}$, we all have $(2\lambda-1)(q^m+1)-\gamma^{\prime}<\delta_1+\lambda\frac{q^{m-h-4}+1}{q+1}+\lambda(q^{h+2}-1)q^{m-(h+3)}+k^{\prime}\lambda q^{m-(h+3)}$. On the other hand, 
\begin{align*}
& \delta_1+\lambda\frac{q^{m-h-4}+1}{q+1}+\lambda(q^{h+2}-1)q^{m-(h+3)}+k^{\prime}\lambda q^{m-(h+3)}<n+\gamma^{\prime}\\
\Leftrightarrow &k^{\prime}<\frac{q^{m+1}+q^m-q^{m-1}+q-q^{m-h-4}-q^{h+3}}{(q+1)(q^{m-(h+3)}-1)},
\end{align*}
it is not difficult to obtain
\begin{align*}
k^{\prime}\le (q-1)q^{h+2}<\frac{q^{m+1}+q^m-q^{m-1}+q-q^{m-h-4}-q^{h+3}}{(q+1)(q^{m-(h+3)}-1)}.
\end{align*}
Therefore, for any $q-1<k<q\frac{q^{m-2}+1}{q+1}$, we can find a $t$ with $1\le t\le m-1$ such that $\overline{\gamma q^t}<\gamma^{\prime}$, which means $\gamma^{\prime}$ is not a coset leader.

Above all, we conclude that $\delta_2=\delta_1-\lambda q\frac{q^{m-2}+1}{q+1}$.
\end{itemize}
Combining the above results, we complete the proof of the theorem.
\end{proof}

\section{BCH codes and their dual}

Throughout this section, we let $n=\lambda(q^m+1)$, where $\lambda\mid q-1$. \cite{gongbinkai} proposed the concept of dually-BCH code: A BCH code of length $n$ over $\mathbb{F}_q$ with respect to an $n$-th primitive root of unity $\zeta_n$ is called a dually-BCH code if its dual code is also a BCH code respect to the same $\zeta_n$. We determine the dimensions of several families of BCH codes and improve the lower bounds on their minimum distances in some cases. In particular, some of the BCH codes we constructed are optimal. Additionally, we present a necessary and sufficient condition for a BCH code $\mathcal{C}_{(q,n,\delta,0)}$ to be a dually-BCH code.

\subsection{Parameters of BCH codes of length $\lambda(q^m+1)$}

With the results on the cyclotomic cosets above, we have the following conclusions on dimensions of BCH codes with $n=\lambda(q^m+1)$. Their proofs follow directly from Theorem \ref{th12}--\ref{th7} and are omitted.

\begin{theorem}\label{th15}
Let $m=2$. Then for $3\le \delta\le \lambda q+1$, we have 
\begin{itemize}
\item[(1)] If $q$ is even or $q$ is odd and $\lambda\le \frac{q-1}{2}$, then
$$\mathrm{dim}(\mathcal{C}_{(q,n,\delta,0)})=\lambda(q^2+1)+7-4\delta+4\lfloor\frac{\delta-2}{q}\rfloor;$$
\item[(2)] If $q$ is odd and $\lambda\ge \frac{q+1}{2}$, then
\begin{equation*}
	\mathrm{dim}(\mathcal{C}_{(q,n,\delta,0)})=	\left\{
	\begin{array}{lcl}
	\lambda(q^2+1)+7-4\delta+4\lfloor\frac{\delta-2}{q}\rfloor	, \ \mathrm{if}~ 3\le \delta\le \frac{q^2+3}{2};\\
	\lambda(q^2+1)+9-4\delta-4\lfloor\frac{\delta-2}{q}\rfloor, \ \mathrm{if}~ \frac{q^2+5}{2}\le \delta\le \lambda q+1.
	\end{array} \right.
\end{equation*}
\end{itemize}
\end{theorem}

\begin{theorem}\label{th16}
Let $m\ge 4$ be even. Then for $3\le \delta\le \lambda q^{\frac{m}{2}}+1$, we have 
$$\mathrm{dim}(\mathcal{C}_{(q,n,\delta,0)})=\lambda q^m+\lambda-1-2m(\delta-2-\lfloor\frac{\delta-2}{q}\rfloor).$$
\end{theorem}

\begin{theorem}\label{th17}
Let $m=3$. Then for $3\le \delta\le 2\lambda q+1$, we have
\begin{itemize}
\item[(1)] If $1\le \lambda<\lceil\frac{q}{2}\rceil$, then
$$\mathrm{dim}(\mathcal{C}_{(q,n,\delta,0)})=\lambda(q^3+1)-6\delta+6\lfloor\frac{\delta-2}{q}\rfloor+11;$$
\item[(2)] If $\lambda\ge \lceil\frac{q}{2}\rceil$, then
\begin{equation*}
	\mathrm{dim}(\mathcal{C}_{(q,n,\delta,0)})=	\left\{
	\begin{array}{lcl}
\lambda(q^3+1)-6\delta+6\lfloor\frac{\delta-2}{q}\rfloor+11	, \ \mathrm{if}~ 3\le \delta\le q^2-q+2;\\
\lambda(q^3+1)-6\delta+6\lfloor\frac{\delta-2}{q}\rfloor+15, \ \mathrm{if}~ q^2-q+3\le \delta\le 2\lambda q+1.
	\end{array} \right.
\end{equation*}
\end{itemize}
\end{theorem}

\begin{theorem}\label{th18}
Let $m\ge 5$ be odd. Then for $3\le \delta\le \lambda(q^{\frac{m-1}{2}}+q)+1$, we have
$$\mathrm{dim}(\mathcal{C}_{(q,n,\delta,0)})=\lambda q^m+\lambda-1-2m(\delta-2-\lfloor\frac{\delta-2}{q}\rfloor).$$
\end{theorem}

For the distances of BCH codes with length 
$n=\lambda(q^m+1)$, we provide the following bound, which improves the lower bound on the minimum distance of the code $\mathcal{C}_{(q,n,2\delta+1,n-\delta+1)}$.

​ 
\begin{theorem}\label{th8}
Let $n=\lambda(q^m+1)$, where $\lambda\mid \mathrm{gcd}(q-1,\delta)$. Then the code $\mathcal{C}_{(q,n,2\delta+1,n-\delta+1)}$ has minimal distance $d\ge 2(\delta+1)$.
\end{theorem}

\begin{proof}
Let $\zeta_n$ be the $n$-th primitive root of $\mathbb{F}_{q^{2m}}^{\ast}$. According to the definition, the generator polynomial $f_{(q,n,2\delta+1,n-\delta+1)}$ of this code has the roots $\zeta_n^i$ for $i$ in the set
\begin{align*}
\{n-\delta+1, n-\delta+2,\cdots ,n-1, 0, 1, \cdots ,n+\delta\}.
\end{align*}
It follows from Lemma \ref{lem2} that $\zeta_n^{n-\delta}$ is also the root of the generator polynomial $f_{(q,n,2\delta+1,n-\delta+1)}$. Therefore, by the BCH bound, we deduce that $d\ge 2(\delta+1)$.
\end{proof}

Furthermore, we obtain the following theorem. Throughout, we assume $m\ge 2$ is an integer. If $m=2$ or $m=3$, we require that $q$ and $\lambda$ satisfy the following conditions:
\begin{itemize}
\item[(1)] If $m=2$, then $q$ is even or $q$ is odd and $\lambda\le \frac{q-1}{2};$
\item[(2)] If $m=3$, then $1\le \lambda\le \lceil \frac{q}{2}\rceil$.
\end{itemize}

\begin{theorem}\label{th5}
Let $\lambda\mid \delta$ and $\lambda\le \delta\le \lambda q^{\frac{m}{2}}-1$ for even $m$, $\lambda\le \delta\le \lambda(q^{\frac{m-1}{2}}+q)-1$ for odd $m$, then the parameters of $\mathcal{C}_{(q,n,2\delta+1,n-\delta+1)}$ is 
\begin{align*}
\left[\lambda(q^m+1), \lambda q^m-2m(2\delta-\frac{
\delta}{\lambda}-\lfloor\frac{\delta-1}{q}\rfloor-\lfloor\frac{\delta}{q}\rfloor+\lfloor\frac{\delta-1}{\lambda q}\rfloor)+\lambda-1, d\ge 2(\delta+1)\right]
\end{align*}
and generator polynomial is
\begin{align*}
f_{(q,n,2\delta+1,n-\delta+1)}(x)=
(x-1)\cdot \prod\limits_{\substack{\gamma=1\\\gamma\not\equiv 0\pmod{q}\\\gamma\not\equiv 0\pmod{\lambda}}}^{\delta}f_{\gamma}(x)f_{-\gamma}(x)\cdot \prod\limits_{\substack{\gamma=1\\\gamma\not\equiv 0\pmod{q}\\\gamma\equiv0\pmod{\lambda}}}^{\delta}f_{\gamma}(x),
\end{align*}
where $f_{\gamma}(x)$ is the minimal polynomial of $\zeta_n^{\gamma}$ over $\mathbb{F}_q$.
\end{theorem}

\begin{proof}
The defining set of $\mathcal{C}_{(q,n,2\delta+1,n-\delta+1)}$ is 
\begin{align*}
T=C_{n-\delta+1}\cup \cdots C_{n-1}\cup C_0\cup C_2\cup \cdots C_{\delta}.
\end{align*}
The conclusion on the dimension and generator polynomial then follow from the definition of BCH codes and Theorem \ref{th12}--\ref{th7}. The lower bound on the minimal distance follows from Theorem \ref{th8}.
\end{proof}

\begin{example}
Take $q=3$ and $m=3$, then $n=56$. By Theorem \ref{th5}, the code $\mathcal{C}_{(3,56,9,53)}$ has parameters $[56,31,d\ge 10]$. By Magma program, we have the minimal distance of $\mathcal{C}_{(3,56,9,53)}$ is $10$. Further, it is the best possible for cyclic code by \cite{dingcunsheng}.
\end{example}

\subsection{The dual codes of BCH codes of length $\lambda(q^m+1)$}

 Let $T$ and $T^{\perp}$ denote the defining sets of $\mathcal{C}_{(q,n,\delta,0)}$ and $\mathcal{C}_{(q,n,\delta,0)}^{\perp}$ with respect to $\beta$ respectively. By definition, $T=C_0\cup C_1\cup\cdots \cup C_{\delta-2}$, and it is well-known that $T^{\perp}=\mathbb{Z}_n\setminus T^{-1}$, where $T^{-1}=\{-t\pmod{n}\mid t\in T\}$. Recall  that $\mathcal{C}_{(q,n,\delta,0)}$ is a dually-BCH code if and only if $T^{\perp}$ is a union of consecutive $q$-cyclotomic cosets. We next establish a sufficient and necessary condition on $\delta$ for $\mathcal{C}_{(q,n,\delta,0)}$ to be a  dually-BCH code. For $\lambda=1$, this case was investigated in \cite{fuyuqing}, so we focus on the case $\lambda>1$.

\begin{proposition}
Let $n=(q-1)(q^m+1)$, where $m$ is an odd integer. If $\gamma$ satisfies $\frac{q^m+1}{q+1}(q^2-2q-1)<\gamma<(q-2)(q^m+1)$, then $\gamma$ is not a coset leader.
\end{proposition}

\begin{proof}
It is not difficult to see that an integer $\gamma$ satisfying $$\frac{q^m+1}{q+1}(q^2-2q-1)<\gamma<(q-2)(q^m+1)$$ 
is a coset leader if and only if one of the following conditions holds:
\begin{itemize}
\item[(1)] $q=3$ and $\gamma$ is even;
\item[(2)] $q=4$, $\gamma\equiv 0\pmod{3}$ or $\gamma\equiv 2\pmod{3}$ ($\gamma\equiv 0\pmod{3}$ can only in the range $\frac{7}{5}(4^m+1)+2\le \gamma\le \frac{3}{2}\cdot 4^m)$;
\item[(3)] $q>4$ and $\gamma\equiv q-4\pmod{q-1}$.
\end{itemize}
We only prove the case $q>4$, the cases $q=3$ and $q=4$ are similar.

Let $q>4$, for any $\gamma$ satisfying  $\frac{q^m+1}{q+1}(q^2-2q-1)<\gamma<(q-2)(q^m+1)$ and $\gamma\equiv q-3\pmod{q-1}$, we now show that $\overline{\gamma q}<\gamma$, which implies that $\gamma$ is not a coset leader. We write $\gamma=\frac{q^m+1}{q+1}(q^2-2q-1)+k(q-1)$ for some integer $k$ with $1\le k<\frac{q^m+1}{q+1}$. Assume that 
$$\frac{q^m+1}{q+1}(q^2-3)=\frac{q^m+1}{q+1}(q^2-2q-1)\cdot q-an,$$
where $a$ is a positive integer. It is straightforward to verify that
$$\frac{q^m+1}{q+1}(q^2-3)<\gamma q-an<n+\gamma.$$
This implies $\overline{\gamma q}<\gamma$, completing the proof of the proposition.
\end{proof}

\begin{theorem}\label{th11}
Let $n=\lambda(q^m+1)$, where $m\ge 3$ is odd and $\lambda\ge 2$. Then for $3\le \delta\le \delta_1+1$, $\mathcal{C}_{(q,n,\delta,0)}$ is a dually-BCH code if and only if 
\begin{align*}
\delta_2+2\le \delta\le \delta_1+1.
\end{align*}
\end{theorem}

\begin{proof}
We split the proof into two cases depending on whether $\frac{q-1}{\lambda}$ is odd or even.
\begin{itemize}
\item[$\mathbf{Case~1}:$] If $\frac{q-1}{\lambda}$ is even. Notice that $q$ is must power of an odd prime and $\lambda\le \frac{q-1}{2}$. Reviewing that $C_{\delta_1}=\{\frac{2\lambda-1}{2}(q^m+1)\}$, $C_{\delta_2}=\{\frac{q^m+1}{q+1}(\lambda q-\frac{q+1}{2}),\frac{q^m+1}{q+1}(\lambda q+2\lambda-\frac{q+1}{2})\}$, it follows that $C_{n-\delta_1}=\{\frac{q^m+1}{2}\}$, $C_{n-\delta_2}=C_{\frac{q^m+1}{q+1}(\frac{q+1}{2}-\lambda)}=\{\frac{q^m+1}{q+1}(\frac{q+1}{2}-\lambda), \frac{q^m+1}{q+1}(\frac{q+1}{2}+\lambda)\}$.

When $3\le \delta\le \frac{q^m+1}{2}$. We have $C_1, C_{\delta_1}\subset T^{\perp}$ and $(\lambda-1)(q^m+1)+1\not\in T^{\perp}$. Since the maxmum element in $C_1$ is $q^m$ and 
$$q^m<(\lambda-1)(q^m+1)+1<\delta_1,$$
then we can conclude that $\mathcal{C}_{(q,n,\delta,0)}$ is not a dually-BCH code.

When $\frac{q^m+1}{2}\le \delta\le \delta_2+1$. Consider the $q$-cyclotomic coset $C_{\frac{q-1+2\lambda}{2}}=C_{n-\frac{q-1+2\lambda}{2}}\nsubseteq T^{\perp}$, we have 
$$\frac{q^m+1}{2}<\frac{q-1+2\lambda}{2}q^{m-1}<\frac{q^m+1}{q+1}(\frac{q+1}{2}+\lambda).$$
Consider the $q$-cyclotomic coset $C_{\frac{(q-1)(q^m+1)}{2(q+1)}}=C_{n-\frac{(q-1)(q^m+1)}{2(q+1)}}$, we have 
$$\frac{q^m+1}{q+1}(\frac{q-1}{2}-\lambda)<\frac{(q-1)(q^m+1)}{2(q+1)}<\frac{q^m+1}{2}.$$
It implies that $\mathcal{C}_{(q,n,\delta,0)}$ is not a dually-BCH code.

\item[$\mathbf{Case~2}:$] If $\frac{q-1}{\lambda}$ is odd, we split the proof into two subcases:  $1<\lambda<q-1$ and $\lambda=q-1$. Reviewing that $C_{\delta_1}=\{\frac{q^m+1}{q+1}(\frac{2\lambda q+\lambda-q-1}{2}), \frac{q^m+1}{q+1}(\frac{2\lambda q+3\lambda-q-1}{2})\}$, it follows that $C_{n-\delta_1}=\{\frac{q^m+1}{q+1}(\frac{q+1-\lambda}{2}), \frac{q^m+1}{q+1}(\frac{q+1+\lambda}{2})\}$.
\begin{itemize}
\item[$\mathbf{Case~2.1}:$] $1<\lambda<q-1$. When $3\le \delta\le \frac{q^m+1}{q+1}(\frac{q+1-\lambda}{2})+1$, we have $C_1, C_{\delta_1}\subset T^{\perp}$ and $(\lambda-1)(q^m+1)+1\not\in T^{\perp}$. Since the maxmum element in $C_1$ is $q^m$ and 
$$q^m<(\lambda-1)(q^m+1)+1<\delta_1, $$ 
then we can conclude that $\mathcal{C}_{(q,n\delta,0)}$ is not a dually-BCH code.

When $\frac{q^m+1}{q+1}(\frac{q+1-\lambda}{2})+2\le \delta\le (\lambda-1)(q^m+1)+1$, we have $C_{q-1}\nsubseteq T^{\perp}$, $C_{q^m+1}$, $C_{n-\delta_1}\subseteq T^{\perp}$. According to 
$$\frac{q^m+1}{q+1}(\frac{q+1+\lambda}{2})<(q-1)q^{m-1}<q^m+1,$$
then we have $\mathcal{C}_{(q,n,\delta,0)}$ is not a dually-BCH code.

When $(\lambda-1)(q^m+1)+2\le \delta\le \delta_2+1$. Since the coset leaders $\gamma$ satisfying $\gamma>(\lambda-1)(q^m+1)$ all satisfy $\gamma\equiv \lambda-1\pmod{\lambda}$, then according to $T^{\perp}={\mathbb{Z}_n\setminus T}^{-1}$, we have all elements $\gamma^{\prime}\in T^{\perp}$ satisfies $\lambda-1\pmod{\lambda}$ and $\mid T^{\perp}\mid>2$, which implies that $\mathcal{C}_{(q,n,\delta,0)}$ is not a dually-BCH code.
\item[$\mathbf{Case~2.2}:$] $\lambda=q-1$. When $3\le \delta\le \frac{q^m+1}{q+1}+1$, we have $C_1,C_{\delta_1}\subset T^{\perp}$ and $(q-2)(q^m+1)+1\not\in T^{\perp}$. Since tha maxmum element in $C_1$ is $q^m$ and 
$$q^m<(q-2)(q^m+1)+1<\delta_1,$$
then we can conclude that $\mathcal{C}_{(q,n,\delta,0)}$ is not a dually-BCH code.

When $\frac{q^m+1}{q+1}+1\le \delta\le \frac{q^m+1}{q+1}(\frac{q^2-2q-1}{2})+1$, we have $C_1, C_{2\frac{q^m+1}{q+1}}\subseteq T^{\perp}$, $C_{2(q-1)}\nsubseteq T^{\perp}$, notide that $C_{2\frac{q^m+1}{q+1}}=\{2\frac{q^m+1}{q+1}, 2q\frac{q^m+1}{q+1}\}$. The smallest integer in $C_1$ that is larger than $2\frac{q^m+1}{q+1}$ is $(q-1)q^{m-1}+1$, the largest integer in $C_1$ that is smaller than $2\frac{q^m+1}{q+1}$ is $q^{m-1}$. For $2\frac{q^m+1}{q+1}$, we have 
$$q^{m-1}<2(q-1)q^{m-2}<2\frac{q^m+1}{q+1}$$ 
and 
$$2\frac{q^m+1}{q+1}<(q-1)q^{m-1}<(q-1)q^{m-1}+1.$$
For $2q\frac{q^m+1}{q+1}$, we have 
$$q^m<2(q-1)q^{m-1}<2q\frac{q^m+1}{q+1}.$$
From this, we can conclude that $\mathcal{C}_{(q,n,\delta,0)}$ is not a dually-BCH code.

When $\frac{q^m+1}{q+1}(q^2-2q-1)+2\le \delta\delta_2+1$. Since the coset leaders $\gamma$ satisfying $\gamma>(\lambda-1)(q^m+1)$ all satisfy $\gamma\equiv \lambda-1\pmod{\lambda}$, then according to $T^{\perp}={\mathbb{Z}_n\setminus T}^{-1}$, we have all elements $\gamma^{\prime}\in T^{\perp}$ satisfies $\lambda-1\pmod{\lambda}$ and $\mid T^{\perp}\mid>2$, which implies that $\mathcal{C}_{(q,n,\delta,0)}$ is not a dually-BCH code.
\end{itemize}
\end{itemize}
This completes the proof of the theorem.
\end{proof}

\section{The Enumeration of LCD cyclic codes of length $ \lambda(q^{m}+1)$ over $\mathbb{F}_q$}
This section is devoted to the enumeration of all the LCD cyclic codes of length $n = \lambda(q^{m}+1)$ over $\mathbb{F}_{q}$, where $\lambda \mid q-1$. As a preparation, in Subsection $3.1$ we present a formula counting the number of the $q$-cyclotomic cosets modulo an arbitrary positive integer with a fixed size. In Subsection $3.2$, we give the exact enumeration of the LCD cyclic codes of length $n$, based on the results about $q$-cyclotomic cosets modulo $n$.

\subsection{Enumeration of cyclotomic cosets}
In this subsection we let $n$ be a positive integer coprime to $q$. For a positive integer $\tau$, we denote by $N_{\tau}$ the number of the $q$-cyclotomic cosets modulo $n$ with size $\tau$. It is clear that there are only finitely many $\tau$ such that $N_{\tau} > 0$, which are given by the lemma below.

\begin{lemma}\label{lem4}
	The values of $\tau$ for which $N_{\tau} > 0$ are
	$$\tau = \mathrm{ord}_{d}(q),$$
	where $d$ ranges over the positive factor of $n$.
\end{lemma}
	
\begin{proof}
	Let $\gamma \in \mathbb{Z}/n\mathbb{Z}$, and $c_{n/q}(\gamma)$ be the $q$-cyclotomic coset modulo $n$ containing $\gamma$. The size $\tau$ of $c_{n/q}(\gamma)$ is the smallest  positive integer such that 
	$$\gamma q^{\tau} \equiv \gamma \pmod{n}.$$
	Equivalently, $\tau = \mathrm{ord}_{d}(q)$ where $d = \frac{n}{\mathrm{gcd}(\gamma,n)}$.
	
	Conversely, for any positive factor $d$ of $n$, it is easy to check that $\mathrm{ord}_{d}(q)$ is the size of the coset $c_{n/q}(\frac{n}{d})$.
\end{proof}

From now on, we assume that $\tau$ is of the form $\tau = \mathrm{ord}_{d}(q)$, where $d$ is a positive factor of $n$. The number $N_{\tau}$ of the $q$-cyclotomic cosets modulo $n$ with size $\tau$ is given in Theorem \ref{thm 1}. Recall that the M\"{o}bius function is defined as
\begin{equation*}
	\mu(m)=	\left\{
	\begin{array}{lcl}
		1, \ \mathrm{if} \ m=1;\\
		(-1)^{r}, \ \mathrm{if} \ m=p_{1}\cdots p_{r} \ \mathrm{where} \ p_{1},\cdots,p_{r} \ \mathrm{are} \ \mathrm{distinct} \ \mathrm{primes};\\
		0, \ \mathrm{otherwise}.
	\end{array} \right.
\end{equation*}

\begin{theorem}\label{thm 1}
	Let the notations be defined as above. Then 
	$$N_{\tau}=\frac{1}{\tau}\sum\limits_{\epsilon\mid \tau}\mu(\epsilon)\cdot \mathrm{gcd}(n,q^{\frac{\tau}{\epsilon}}-1).$$
\end{theorem}

\begin{proof}
	We first count the elements $\gamma \in \mathbb{Z}/n\mathbb{Z}$ such that $|c_{n/q}(\gamma)| = \tau$. A $q$-cyclotomic coset $c_{n/q}(\gamma)$ modulo $n$ has size dividing $\tau$ if and only if 
	$$\gamma q^{\tau} \equiv \gamma \pmod{n},$$
	which is equivalent to
	\begin{equation}\label{eq 2}
		\frac{n}{\mathrm{gcd}(n,\gamma)} \mid \mathrm{gcd}(n,q^{\tau}-1).
	\end{equation}
	On the other hand, for any positive integer $d \mid n$, the number of the integers $0 \leq \gamma \leq n-1$ such that $\mathrm{gcd}(\gamma,n)=d$ is $\varphi(\frac{n}{d})$. Therefore the number of the elements $\gamma \in \mathbb{Z}/n\mathbb{Z}$ satisfying \eqref{eq 2} is
	$$\sum_{\frac{n}{d} \mid \mathrm{gcd}(n,q^{\tau}-1)}\varphi(\frac{n}{d}) = \sum_{d \mid \mathrm{gcd}(n,q^{\tau}-1)}\varphi(d) = \mathrm{gcd}(n,q^{\tau}-1).$$
	Furthermore, assume that $\mathcal{P} = \{p_{1},\cdots,p_{s}\}$ is the set of all distinct prime factors of $\tau$. For any nonempty set $\mathcal{U} \subseteq \mathcal{P}$, set 
	$$\tau_{\mathcal{U}} = \frac{\tau}{\prod\limits_{u \in \mathcal{U}}u}.$$
	By the same argument as above, the number of the elements $\gamma \in \mathbb{Z}/n\mathbb{Z}$ such that $|c_{n/q}(\gamma)| \mid \tau_{\mathcal{U}}$ is $\mathrm{gcd}(n,q^{\tau_{\mathcal{U}}}-1)$. Hence applying the inclusion-exclusion formula yields that the number of the $\gamma \in \mathbb{Z}/n\mathbb{Z}$ such that $|c_{n/q}(\gamma)| = \tau$ is given by
    \begin{align*}\label{eq 3}
		&\mathrm{gcd}(n,q^{\tau}-1) - \sum_{\substack{\mathcal{U}\subset \mathcal{P}\\\mid \mathcal{U}\mid=1}}\mathrm{gcd}(n,q^{\tau_{\mathcal{U}}}-1) + \sum_{\substack{\mathcal{U}\subset \mathcal{P}\\\mid \mathcal{U}\mid=2}}\mathrm{gcd}(n,q^{\tau_{\mathcal{U}}}-1) - \cdots + (-1)^{s}\mathrm{gcd}(n,q^{\tau_{\mathcal{P}}}-1)\\
		&= \sum_{\epsilon \mid \tau} \mu(\epsilon)\cdot \mathrm{gcd}(n,q^{\frac{\tau}{\epsilon}}-1).
	\end{align*}

    Notice that different $q$-cyclotomic cosets modulo $n$ are disjoint. Then the number of the $q$-cyclotomic cosets modulo $n$ with size $\tau$ is
    $$N_{\tau} = \frac{1}{\tau}\sum_{\epsilon \mid \tau} \mu(\epsilon)\cdot \mathrm{gcd}(n,q^{\frac{\tau}{\epsilon}}-1).$$
\end{proof}

In the following part of this subsection, we focus on the case that $n = \lambda(q^{m}+1)$, where $\lambda$ is a divisor of $q-1$.

\begin{theorem}\label{th2}
	Let $n = \lambda(q^{m}+1)$, where $\lambda \mid q-1$ and $m$ is a positive integer. Write $m = 2^{v}m_{0}$, where $v = v_{2}(m) \geq 0$, $2\nmid m_0$, and $h = \mathrm{gcd}(2,\frac{q-1}{\lambda})$.
	\begin{itemize}
		\item[(1)]
		If both $q$ and $m$ are odd, the possible values of $\tau$ for which $N_{\tau} > 0$ are $\tau =1$, $2$, or $2\tau_0$, where $\tau_{0} > 1$ is a divisor of $m_{0}$. Moreover, we have
		\begin{equation*}
			N_{\tau}=	\left\{
			\begin{array}{lcl}
				\lambda h, \ \tau=1;\\
				\frac{\lambda}{2}(q-h+1), \ \tau=2;\\
				\frac{\lambda}{\tau}\sum\limits_{\epsilon\mid \tau_0}\mu(\epsilon)\cdot q^{\frac{\tau}{2\epsilon}}, \ \tau=2\tau_0 \ \mathrm{for} \ 1<\tau_0\mid m_0 .
			\end{array} \right.
		\end{equation*}
		\item[(2)]
		If $q$ is odd and $m$ is even, the possible values of $\tau$ for which $N_{\tau} > 0$ are $\tau =1$, $2$, $2^{v+1}$, or $2^{v+1}\tau_0$, where $\tau_{0} > 1$ is a divisor of $m_{0}$. Moreover, we have
		\begin{equation*}
			N_{\tau}=	\left\{
			\begin{array}{lcl}
				\lambda h, \ \tau=1;\\
				\frac{\lambda}{2}(2-h), \ \tau=2;\\
				\frac{\lambda}{\tau}(q^{\frac{\tau}{2}}-1), \ \tau=2^{v+1};\\
				\frac{\lambda}{\tau}\sum\limits_{\epsilon\mid \tau_0}\mu(\epsilon)\cdot q^{\frac{\tau}{2\epsilon}}, \ \tau=2^{v+1}\tau_0 \ \mathrm{for} \ 1<\tau_0\mid m_0.
			\end{array} \right.
		\end{equation*}
		\item[(3)]
		If $q$ is a power of $2$, the possible values of $\tau$ for which $N_{\tau} > 0$ are $\tau =1$ or $2^{v+1}\tau_0$, where $\tau_{0} > 1$ is a divisor of $m_{0}$. Moreover, we have
		\begin{equation*}
			N_{\tau}=	\left\{
			\begin{array}{lcl}
				\lambda , \ \tau=1;\\
				\frac{\lambda}{\tau}\sum\limits_{\epsilon\mid \tau_0}\mu(\epsilon)\cdot q^{\frac{\tau}{2\epsilon}}, \ \tau=2^{v+1}\tau_0 \ \mathrm{for} \ \tau_0\mid m_0.
			\end{array} \right.
		\end{equation*}
	\end{itemize} 
\end{theorem}

\begin{proof}
 Applying Lemma \ref{lem4}, we know that all possible values of $\tau$ are divisors of $\mathrm{ord}_n(q)=2m$. First, we assume that $q$ is odd, then we treat the following cases seperately. 

\begin{itemize}
\item[$\mathbf{Case~1: }$]For $\tau=1$, it is straightforward to obtain
\begin{align*}
N_{1}&=\mathrm{gcd}(n,q-1)\\
&=\lambda\cdot \mathrm{gcd}(q^m+1,\frac{q-1}{\lambda})=\lambda h.
\end{align*}

\item[$\mathbf{Case~2:}$]

For $1<\tau\mid m$, assume $\tau=2^u\cdot \tau_0$, where $u=v_2(\tau)\ge 0$ and $\mathrm{gcd}(2,\tau_0)=1$, then 
\begin{align*}
N_{\tau}&=\frac{1}{\tau}\sum\limits_{\epsilon\mid \tau}\mu(\epsilon)\cdot \mathrm{gcd}(n, q^{\frac{\tau}{\epsilon}}-1)\\
&=\frac{\lambda}{\tau}\sum\limits_{\epsilon\mid \tau}\mu(\epsilon)\cdot \mathrm{gcd}(q^m+1,\frac{q^{\frac{\tau}{\epsilon}}-1}{\lambda})\\
&=\frac{\lambda}{\tau}\sum\limits_{\epsilon\mid \tau}\mu(\epsilon)\cdot \mathrm{gcd}(2,\frac{q^{\frac{\tau}{\epsilon}}-1}{\lambda}).
\end{align*}

\begin{itemize}
\item[$\mathbf{Case~2.1:}$]

If $\tau>1$ is odd, for any $\epsilon\mid \tau$, $\frac{\tau}{\epsilon}$ is odd, then the lifting-the-exponent lemmas imply that
\begin{align*}
N_{\tau}&=\frac{\lambda}{\tau}\sum\limits_{\epsilon\mid \tau}\mu(\epsilon)\cdot \mathrm{gcd}(2,\frac{q-1}{\lambda})\\
&=\frac{\lambda h}{\tau}\sum\limits_{\epsilon\mid \tau}\mu(\epsilon)=0.
\end{align*}

\item[$\mathbf{Case~2.2:}$]
If $4\mid \tau$, then
\begin{align*}
N_{\tau}&=\frac{\lambda}{\tau}(\sum\limits_{\epsilon\mid \tau_0}\mu(\epsilon)+\sum\limits_{\substack{\epsilon\mid 2\tau_0\\ 2\mid \epsilon}}\mu(\epsilon))\cdot \mathrm{gcd}(2,\frac{q^{\frac{\tau}{\epsilon}}-1}{\lambda})\\
&=\frac{\lambda}{\tau}\sum\limits_{\epsilon\mid \tau_0}\mu(\epsilon)\cdot (\mathrm{gcd}(2,\frac{q^{\frac{\tau}{\epsilon}}-1}{\lambda})-\mathrm{gcd}(2,\frac{q^{\frac{\tau}{2\epsilon}}-1}{\lambda})).
\end{align*}
Remembering that for any $\epsilon\mid \tau_0$, $\frac{\tau}{\epsilon}$ and $\frac{\tau}{2\epsilon}$ are even, then the lifting-the-exponent lemmas imply that 
\begin{align*}
N_{\tau}=\frac{\lambda}{\tau}\sum\limits_{\epsilon\mid \tau_0}\mu(\epsilon)\cdot (2-2)=0.
\end{align*}

\item[$\mathbf{Case~2.3:}$]
If $\tau=2\tau_0$, then

\begin{align*}
N_{\tau}&=\frac{\lambda}{\tau}\sum\limits_{\epsilon\mid \tau_0}\mu(\epsilon)\cdot (\mathrm{gcd}(2,\frac{q^{\frac{\tau}{\epsilon}}-1}{\lambda})-\mathrm{gcd}(2,\frac{q^{\frac{\tau}{2\epsilon}}-1}{\lambda})).
\end{align*}

Notice that for any $\epsilon\mid \tau_0$, $\frac{\tau}{\epsilon}$ is even and $\frac{\tau}{2\epsilon}$ is odd, then the lifting-the-exponent lemmas imply that 

\begin{align*}
N_{\tau}=\frac{\lambda}{\tau}\sum\limits_{\epsilon\mid \tau_0}\mu(\epsilon)(2-h).
\end{align*}

It follows that when $m$ is even, we have 

\begin{align*}
N_2=\frac{\lambda}{2}(2-h)
\end{align*}

and 

\begin{align*}
N_{\tau}=\frac{\lambda}{\tau}(2-h)\sum\limits_{\epsilon\mid \tau_0}\mu(\epsilon)=0
\end{align*}
for $\tau=2\tau_0$, where $1<\tau_0\mid m_0$.
\end{itemize}

\item[$\mathbf{Case~3:}$]

For $\tau=2^{v+1}\tau_0$, where $1\le \tau_0\mid m_0$, then

\begin{align*}
N_{\tau}&=\frac{1}{\tau}\sum\limits_{\epsilon\mid \tau}\mu(\epsilon)\cdot \mathrm{gcd}(n,q^{\frac{\tau}{\epsilon}}-1)\\
&=\frac{1}{\tau}(\sum\limits_{\epsilon\mid \tau_0}\mu(\epsilon)+\sum\limits_{\substack{\epsilon\mid 2\tau_0\\2\mid \epsilon}}\mu(\epsilon))\cdot \mathrm{gcd}(n,q^{\frac{\tau}{\epsilon}}-1)\\
&=\frac{1}{\tau}\sum\limits_{\epsilon\mid \tau_0}\mu(\epsilon)(\mathrm{gcd}(n,q^{\frac{\tau}{\epsilon}}-1)-\mathrm{gcd}(n,q^{\frac{\tau}{2\epsilon}}-1))\\
&=\frac{1}{\tau}\sum\limits_{\epsilon\mid \tau_0}\mu(\epsilon)(\lambda(q^{\frac{\tau}{2\epsilon}}+1)\cdot \mathrm{gcd}(\frac{q^m+1}{q^{\frac{\tau}{2\epsilon}}+1},\frac{q^{\frac{\tau}{2\epsilon}}-1}{\lambda})-\lambda\mathrm{gcd}(q^m+1,\frac{q^{\frac{\tau}{2\epsilon}}-1}{\lambda})).
\end{align*}

\begin{itemize}
\item[$\mathbf{Case~3.1:}$]
When $m$ is odd, that is $\tau=2\tau_0$, where $1\le \tau_0\mid m_0$. By the lifting-the-exponent lemmas, we find that for any $\epsilon\mid \tau_0$, 
\begin{align*}
v_2(q^m+1)=v_2(q+1)=v_2(q^{\frac{\tau}{2\epsilon}}+1)
\end{align*}
and 
\begin{align*}
v_2(q^{\frac{\tau}{2\epsilon}}-1)=v_2(q-1)
\end{align*}
hold. Then it follows that 
\begin{align*}
N_{\tau}=\frac{\lambda}{\tau}\sum\limits_{\epsilon\mid \tau_0}\mu(\epsilon)\cdot q^{\frac{\tau}{2\epsilon}}+\frac{\lambda}{\tau}(1-h)\sum\limits_{\epsilon\mid \tau_0}\mu(\epsilon).
\end{align*}
That means 
$$N_2=\frac{\lambda}{2}q+\frac{\lambda}{2}(1-h)=\frac{\lambda}{2}(q+1-h)$$
and 
$$N_{\tau}=\frac{\lambda}{\tau}\sum\limits_{\epsilon\mid \tau_0}\mu(\epsilon)\cdot q^{\frac{\tau}{2\epsilon}}$$
for $\tau=2\tau_0$, where $1<\tau_0\mid m_0$.

\item[$\mathbf{Case~3.2:}$]
When $m$ is even, that is $\tau=2^{v+1}\tau_0$, where $1\le \tau_0\mid m_0$ and $v\ge 1$. For any $\epsilon\mid \tau_0$, also by the lift-the-exponent lemmas, we have 
$$v_2(q^m+1)=v_2(q^{\frac{\tau}{2\epsilon}}+1)=1$$
and 
$$v_2(q^{\frac{\tau}{2\epsilon}}-1)=v_2(q-1)+v_2(\frac{\tau}{2\epsilon})\ge v_2(q-1)+1.$$
Then it follows that
\begin{align*}
N_{\tau}&=\frac{1}{\tau}\sum\limits_{\epsilon\mid \tau_0}\mu(\epsilon)(\lambda(q^{\frac{\tau}{2\epsilon}}+1)-2\lambda)\\
&=\frac{\lambda}{\tau}\sum\limits_{\epsilon\mid \tau_0}\mu(\epsilon)\cdot q^{\frac{\tau}{2\epsilon}}-\frac{\lambda}{\tau}\sum\limits_{\epsilon\mid \tau_0}\mu(\epsilon).
\end{align*}
That means, 
$$N_{2^{v+1}}=\frac{\lambda}{\tau}\sum\limits_{\epsilon\mid \tau_0}\mu(\epsilon)(q^{\frac{\tau}{2\epsilon}}-1)$$
and 
$$N_{\tau}=\frac{\lambda}{\tau}\sum\limits_{\epsilon\mid \tau_0}\mu(\epsilon)\cdot q^{\frac{\tau}{2\epsilon}}$$
for $\tau=2^{v+1}\tau_0$, where $1<\tau_0\mid m_0$.
\end{itemize}
\end{itemize}

Next, we turn to the even $q$ case. For $\tau=1$, it is clear that
\begin{align*}
N_{1}=\mathrm{gcd}(n,q-1)=\lambda\cdot \mathrm{gcd}(q^m+1,\frac{q-1}{\lambda})=\lambda. 
\end{align*}
For $1<\tau\mid m$, we have 
\begin{align*}
N_{\tau}&=\frac{1}{\tau}\sum\limits_{\epsilon\mid \tau}\mu(\epsilon)\cdot \mathrm{gcd}(n,q^{\frac{\tau}{\epsilon}}-1)\\
&=\frac{\lambda}{\tau}\sum\limits_{\epsilon\mid \tau}\mu(\epsilon)\cdot \mathrm{gcd}(q^m+1,\frac{q^{\frac{\tau}{\epsilon}}-1}{\lambda})\\
&=\frac{\lambda}{\tau}\sum\limits_{\epsilon\mid \tau}\mu(\epsilon)=0.
\end{align*}
For $\tau=2^{v+1}\tau_0$, where $1\le \tau_0\mid m_0$, then 
\begin{align*}
N_{\tau}&=\frac{1}{\tau}\sum\limits_{\epsilon\mid \tau}\mu(\epsilon)\cdot \mathrm{gcd}(n,q^{\frac{\tau}{\epsilon}}-1)\\
&=\frac{\lambda}{\tau}(\sum\limits_{\epsilon\mid \tau_0}\mu(\epsilon)+\sum\limits_{\substack{\epsilon\mid 2\tau_0\\2\mid \epsilon}}\mu(\epsilon))\cdot \mathrm{gcd}(q^m+1,\frac{q^{\frac{\tau}{\epsilon}}-1}{\lambda})\\
&=\frac{\lambda}{\tau}\sum\limits_{\epsilon\mid \tau_0}\mu(\epsilon)(\mathrm{gcd}(q^m+1,\frac{q^{\frac{\tau}{\epsilon}}-1}{\lambda})-\mathrm{gcd}(q^m+1,\frac{q^{\frac{\tau}{2\epsilon}}-1}{\lambda}))\\
&=\frac{\lambda}{\tau}\sum\limits_{\epsilon\mid \tau_0}\mu(\epsilon)((q^{\frac{\tau}{2\epsilon}}+1)\cdot \mathrm{gcd}(\frac{q^m+1}{q^{\frac{\tau}{2\epsilon}}+1},\frac{q^{\frac{\tau}{2\epsilon}}-1}{\lambda})-1)\\
&=\frac{\lambda}{\tau}\sum\limits_{\epsilon\mid \tau_0}\mu(\epsilon)\cdot q^{\frac{\tau}{2\epsilon}}.
\end{align*}
Here we complete the proof.

\end{proof}

\begin{corollary}\label{coro1}
Let the notations be defined as Theorem \ref{th2}. Assume further that $\lambda=q-1$.
	\begin{itemize}
		\item[(1)]
		If both $q$ and $m$ are odd, then
		\begin{equation*}
			N_{\tau}=	\left\{
			\begin{array}{lcl}
				q-1 , \ \tau=1;\\
				\frac{q(q-1)}{2}, \ \tau=2;\\
				\frac{q-1}{\tau}\sum\limits_{\epsilon\mid \tau_0}\mu(\epsilon)\cdot q^{\frac{\tau}{2\epsilon}}, \ \tau=2\tau_0 \ \mathrm{for} \ 1<\tau_0\mid m_0.
			\end{array} \right.
		\end{equation*}
		\item[(2)]
		If $q$ is odd and $m$ is even, then
		\begin{equation*}
			N_{\tau}=	\left\{
			\begin{array}{lcl}
				q-1 , \ \tau=1;\\
				\frac{q-1}{2}, \ \tau=2;\\
				\frac{q-1}{\tau}(q^{\frac{\tau}{2}}-1),\tau=2^{v+1};\\
				\frac{q-1}{\tau}\sum\limits_{\epsilon\mid \tau_0}\mu(\epsilon)\cdot q^{\frac{\tau}{2\epsilon}}, \ \tau=2^{v+1}\tau_0 \ \mathrm{for} \ 1<\tau_0\mid m_0.
			\end{array} \right.
		\end{equation*}
		\item[(3)]
		If $q$ is even, then
		\begin{equation*}
			N_{\tau}=	\left\{
			\begin{array}{lcl}
				q-1 , \ \tau=1;\\
				\frac{q-1}{\tau}\sum\limits_{\epsilon\mid \tau_0}\mu(\epsilon)\cdot q^{\frac{\tau}{2\epsilon}}, \ \tau=2^{v+1}\tau_0 \ \mathrm{for} \ \tau_0\mid m_0.
			\end{array} \right.
		\end{equation*}
	\end{itemize}
\end{corollary}

\begin{corollary}\label{coro2}
Let the notations be defined as Theorem \ref{th2}. Assume further that $\lambda=1$.
	\begin{itemize}
		\item[(1)]
		If $q$ is odd, then
		\begin{equation*}
			N_{\tau}=	\left\{
			\begin{array}{lcl}
				2, \ \tau=1;\\
				\frac{1}{\tau}(q^{2^v}-1), \ \tau=2^{v+1};\\
				\frac{1}{\tau}\sum\limits_{\epsilon\mid \tau_0}\mu(\epsilon)\cdot q^{\frac{\tau}{2\epsilon}}, \ \tau=2^{v+1}\tau_0 \ \mathrm{for} \ 1<\tau_0\mid m_0 .
			\end{array} \right.
		\end{equation*}
		\item[(2)]
		If $q$ is even, then
		\begin{equation*}
			N_{\tau}= \left\{
			\begin{array}{lcl}
				1, \ \tau=1;\\
				\frac{1}{\tau}\sum\limits_{\epsilon\mid \tau_0}\mu(\epsilon)\cdot q^{\frac{\tau}{2\epsilon}}, \ \tau=2^{v+1}\tau_0 \ \mathrm{for} \ \tau_0\mid m_0 .
			\end{array} \right.
		\end{equation*}
	\end{itemize}
\end{corollary}

\subsection{Enumeration of LCD cyclic codes of length $\lambda(q^m+1)$ over $\mathbb{F}_q$}
Let $n$ be a positive integer coprime to $q$, and $\gamma \in \mathbb{Z}/n\mathbb{Z}$, with the $q$-cyclotomic coset $C_{\gamma}$ modulo $n$ containing $\gamma$. Denote by $\mathrm{Ld}(\gamma)$ the coset leader of $C_{\gamma}$. The following result, proved by Li and Ding \cite{lichengju}, presents a general formula for the number of LCD cyclic codes of length $n$ over $\mathbb{F}_{q}$.

\begin{lemma}\cite{lichengju}
	The total number of LCD cyclic codes over $\mathbb{F}_q$ of length $n$ is equal to $2^{|\Pi|}-1$, where
	\begin{align*}
		\Pi=\Gamma\setminus \{\mathrm{max}\{\gamma,\mathrm{Ld}(n-\gamma)\}\mid \gamma
		\in \Gamma, n-\gamma\notin C_{\gamma}\}.
	\end{align*}
\end{lemma}

From now on, we assume that $n=\lambda(q^{m}+1)$, where $\lambda \mid q-1$ and $m$ is a positive integer. To calculate explicitly the number of LCD cyclic codes of length $n$ over $\mathbb{F}_{q}$, we begin with the following lemma.

\begin{lemma}\label{lem1}
	For any $\gamma \in \mathbb{Z}/n\mathbb{Z}$, $\gamma$ and $n-\gamma$ lie in the same $q$-cyclotomic coset modulo $n$ if and only if one of the conditions below holds:
	\begin{itemize}
		\item[(1)] $\lambda \mid \gamma$;
		\item[(2)] $q \equiv 3 \pmod{4}$, both $\lambda$ and $m$ are even, and $\frac{n}{4} \mid \gamma$.
	\end{itemize}
\end{lemma}

\begin{proof}
The sufficiency is obvious, we only need to consider the necessity. For any $\gamma\in \mathbb{Z}/n\mathbb{Z}$, if $n-\gamma$ and $\gamma$ are in the same $q$-cyclotomic coset, then there exists an integer $t$ satisfies 
$$n-\gamma\equiv \gamma q^t\pmod{n}\Leftrightarrow \frac{n}{\mathrm{gcd}(n,q^t+1)}\mid \gamma.$$
When $q$ is even or $q\equiv 1\pmod{4}$ or $q\equiv 3\pmod{4}$ and $m$ is odd, for any positive integer $t$, combine the lift-the-exponent lemmas, we all have 
$$\mathrm{gcd}(n,q^t+1)=\mathrm{gcd}(q^m+1,q^t+1)\mid q^m+1$$
and then
$$\lambda=\frac{n}{\mathrm{gcd}(n,q^m+1)}\mid \frac{n}{\mathrm{gcd}(n,q^t+1)}\mid \gamma.$$
That means in this case, $n-\gamma$ and $\gamma$ are in the same $q$-cyclotomic coset if and only if $\lambda\mid \gamma$.

When $q\equiv 3\pmod{4}$ and $m$ is even. For any odd $t$, it is obvious that 
$$\mathrm{gcd}(n,q^t+1)=2^{v_2(\lambda)+1}.$$
For any even $t$, we have 
$$\mathrm{gcd}(n,q^t+1)=\mathrm{gcd}(q^m+1,q^t+1),$$
and then 
$$\lambda=\frac{n}{\mathrm{gcd}(n,q^m+1)}\mid \frac{n}{\mathrm{gcd}(n,q^t+1)}\mid \gamma.$$
That means in this case, $n-\gamma$ and $\gamma$ are in the same $q$-cyclotomic coset if and only if $\lambda\mid \gamma$ or $\frac{n}{4}\mid \gamma$.

Thus, we complete the proof of this lemma.
\end{proof}

Now we give the enumeration of LCD cyclic codes of length $n = \lambda(q^{m}+1)$ over $\mathbb{F}_{q}$.

\begin{theorem}\label{th3}
	Let $n=\lambda(q^m+1)$, where $\lambda\mid q-1$, $m$ is a positive integer. Set $m=2^vm_0$, where $v = v_{2}(m) \geq 0$ and $2 \nmid m_{0}$. Then the total number of LCD cyclic codes of length $n$ over $\mathbb{F}_{q}$ is $2^{\mid \Pi\mid}-1$, where $|\Pi|$ is given as follows:
	\begin{itemize}
		\item[(1)] Assume that $\mathrm{gcd}(2,\frac{q-1}{\lambda})=1$.
		\begin{itemize}
			\item[(1,1)] If both $q$ and $m$ are odd, then
			\begin{align*}
					|\Pi|=\frac{\lambda(q+2)+q+3}{4}+\sum\limits_{1<\tau_0\mid m}(\frac{\lambda+1}{4\tau_0}\sum\limits_{\epsilon\mid \tau_0}\mu(\epsilon)\cdot q^{\frac{\tau_0}{\epsilon}});
			\end{align*}
			\item[(1.2)]If $q$ is odd and $m$ is even, then
			\begin{equation*}
				|\Pi|= \left\{
				\begin{array}{lcl}
					\frac{3\lambda+4}{4}+\frac{\lambda+1}{2^{v+2}}(q^{2^v}-1)+\sum\limits_{1<\tau_0\mid m_0}(\frac{\lambda+1}{2^{v+2}\tau_0}\sum\limits_{\epsilon\mid \tau_0}\mu(\epsilon)\cdot q^{\frac{2^v\tau_0}{\epsilon}}), \ q\equiv 1\pmod{4};\\
					\frac{3(\lambda+2)}{4}+\frac{\lambda+1}{2^{v+2}}(q^{2^v}-1)+\sum\limits_{1<\tau_0\mid m_0}(\frac{\lambda+1}{2^{v+2}\tau_0}\sum\limits_{\epsilon\mid \tau_0}\mu(\epsilon)\cdot q^{\frac{2^v\tau_0}{\epsilon}}), \ q\equiv 3\pmod{4}.
				\end{array} \right.
			\end{equation*}
		\end{itemize}
		\item[(2)] Assume that $\mathrm{gcd}(2,\frac{q-1}{\lambda})=2$.
		\begin{itemize}
			\item[(2.1)]If both $q$ and $m$ are odd, then
			\begin{align*}
				|\Pi|=\frac{(\lambda+1)(q+3)}{4}+\sum\limits_{1<\tau_0\mid m}(\frac{\lambda+1}{4\tau_0}\sum\limits_{\epsilon\mid \tau_0}\mu(\epsilon)\cdot q^{\frac{\tau_0}{\epsilon}});
			\end{align*}
			\item[(2.2)]If $q$ is odd and $m$ is even, then
			\begin{align*}
				|\Pi|=\lambda+1+\frac{\lambda+1}{2^{v+2}}(q^{\frac{2^{v+1}}{\epsilon}}-1)+\sum\limits_{1<\tau_0\mid m_0}(\frac{\lambda+1}{2^{v+2}\tau_0}\sum\limits_{\epsilon\mid \tau_0}\mu(\epsilon)\cdot q^{\frac{2^v\tau_0}{\epsilon}}).
			\end{align*}
		\end{itemize} 
		\item[(3)]If $q$ is even, then
		\begin{align*}
			|\Pi|=\frac{\lambda+1}{2}+\sum\limits_{\tau_0\mid m_0}(\frac{\lambda+1}{2^{v+2}\tau_0}\sum\limits_{\epsilon\mid \tau_0}\mu(\epsilon)\cdot q^{\frac{2^v\tau_0}{\epsilon}}).
		\end{align*}
	\end{itemize}
\end{theorem}

\begin{proof}
Assume $\mathrm{gcd}(2,\frac{q-1}{\lambda})=1$ and $q,m$ are odd. Notice that $\lambda$ is must be even in this case and the possible values of $\tau$ for $N_{\tau}>0$ are $\tau=1,2$ or $2\tau_0$, where $\tau_0>1$ is a divisor of $m_0$. Remembering that for any $\gamma\in \mathbb{Z}/n\mathbb{Z}$, $\gamma$ and $n-\gamma$ in the same $q$-cyclotomic coset if and only if $\lambda\mid \gamma$ in this case.

For any $\gamma\in \mathbb{Z}/n\mathbb{Z}$ satisfies $\tau=\mid C_{\gamma}\mid=1$, it is obvious to have

\begin{align*}
\gamma q\equiv\gamma\pmod{n}
&\Leftrightarrow \frac{\lambda(q^m+1)}{\lambda\mathrm{gcd}(q^m+1,\frac{q-1}{\lambda})}\mid \gamma\\
&\Leftrightarrow q^m+1\mid \lambda.
\end{align*}
Since $q^m+1\equiv 2\pmod{\lambda}$, then for any $\gamma\in \mathbb{Z}/n\mathbb{Z}$, $\mid C_{\gamma}\mid=1$,  $C_{\gamma}=C_{n-\gamma}$ if and only if $\gamma=\frac{\lambda}{2}(q^m+1)$ or $\gamma=0$. 

For any $\gamma\in \mathbb{Z}/n\mathbb{Z}$ satisfies $\tau=\mid C_{\gamma}\mid=2$, we have 
\begin{align*}
\gamma q^2\equiv \gamma\pmod{n}
&\Leftrightarrow \frac{\lambda(q^m+1)}{\lambda(q+1)\mathrm{gcd}(\frac{q^m+1}{q+1},\frac{q-1}{\lambda})}\mid \gamma\\
&\Leftrightarrow \frac{q^m+1}{q+1}\mid \gamma.
\end{align*}
Since $\frac{q^m+1}{q+1}\equiv 1\pmod{\lambda}$, then for any $\gamma\in \mathbb{Z}/n\mathbb{Z}$, $\mid C_{\gamma}\mid =2$, $C_{\gamma}=C_{n-\gamma}$ if and only if 
$$\gamma\in \{\lambda\frac{q^m+1}{q+1} i\mid 1\le i\le q, i\ne \frac{q+1}{2}\}.$$
Moreover, the elements in this set form $\frac{q-1}{2}$ cyclotomic cosets pairwise.

For any $\gamma\in \mathbb{Z}/n\mathbb{Z}$ satisfies $\tau=\mid C_{\gamma}\mid=2\tau_0$, where $1<\tau_0\mid m_0$, we have 
\begin{align*}
\gamma q^{2\tau_0}\equiv \gamma\pmod{n}&\Leftrightarrow \frac{\lambda(q^m+1)}{\lambda(q^{\tau_0}+1)\cdot \mathrm{gcd}(\frac{q^m+1}{q^{\tau_0}+1},\frac{q^{\tau_0}-1}{\lambda})}\mid \gamma\\
&\Leftrightarrow \frac{q^m+1}{q^{\tau_0}+1}\mid \gamma.
\end{align*}
Since $\frac{q^m+1}{q^{\tau_0}+1}\equiv \pmod{\lambda}$, then for any $\gamma\in \mathbb{Z}/n\mathbb{Z}$, $\mid C_{\gamma}\mid =2\tau_0$, $C_{\gamma}=C_{n-\gamma}$ if and only if 
$$\gamma\in \{\lambda\frac{q^m+1}{q+1}i\mid 1\le i\le q^{\tau_0}, \lambda\frac{q^m+1}{q+1}\nmid \gamma\}.$$
Furthermore, every $2\tau_0$ elements in this set form a cyclotomic coset and there are $ \frac{1}{2\tau_0}\sum\limits_{\epsilon\mid \tau_0}\mu(\epsilon)\cdot q^{\frac{\tau_0}{\epsilon}}$ such cyclotomic cosets in total.

Combining the above discussion with Theorem \ref{th2}, we conclude that
\begin{align*}
\mid\Pi\mid=&2+\frac{\lambda-2}{2}+\frac{q-1}{2}+\frac{\lambda q-(q-1)}{4}+\sum\limits_{1<\tau_0\mid m_0} \frac{1}{2\tau_0}\sum\limits_{\epsilon\mid \tau_0}\mu(\epsilon)\cdot q^{\frac{\tau_0}{\epsilon}}+\sum\limits_{1<\tau_0\mid m_0} \frac{\lambda-1}{4\tau_0}\sum\limits_{\epsilon\mid \tau_0}\mu(\epsilon)\cdot q^{\frac{\tau_0}{\epsilon}}\\
=&\frac{\lambda(q+2)+q+3}{4}+\sum\limits_{1<\tau_0\mid m_0}(\frac{\lambda+1}{4\tau_0}\sum\limits_{\epsilon\mid \tau_0}\mu(\epsilon)\cdot q^{\frac{\tau_0}{\epsilon}}).
\end{align*}
The other cases are similar and we omit it.
\end{proof}

\begin{corollary}\label{coro3}
	Let $n=(q-1)(q^m+1)$, where $q$ is odd and $m$ is an odd prime. Then the total number of LCD cyclic codes of length $n$ over $\mathbb{F}_q$ is equal to 
	$$2^{\frac{{(q+1)}^2}{4}+\frac{q^2(q^{m-1}-1)}{4m}}-1.$$
\end{corollary}

\begin{corollary}\label{coro4}
	Let $n=(q-1)(q^m+1)$, where $q$ is even and $m$ is an odd prime. Then the total number of LCD cyclic codes of length $n$ over $\mathbb{F}_q$ is equal to  
	$$2^{\frac{q^2(q^{m-1}+m-1)+2mq}{4m}}-1.$$
\end{corollary}

\begin{example}
Let $(n,q)=(56,3)$. The number of LCD cyclic codes of length $n=56$ is 
\begin{align*}
2^{\frac{{(q+1)}^2}{4}+\frac{q^2(q^{m-1}-1)}{4m}}-1=2^{10}-1.
\end{align*}
By Magma, we have all $3$-cyclotomic coset modulo $56$ are
\begin{align*}
C_0, C_1, C_2, C_4, C_5, C_7, C_8, C_{10}, C_{11}, C_{14}, C_{28}, C_{29}, C_{35}.
\end{align*}
Among them, only $\gamma=0,2,4,8,10,14,28$ satisfy $C_{\gamma}=C_{n-\gamma}$. Then the number of LCD cyclic codes is 
\begin{align*}
2^{7+\frac{13-7}{2}}-1=2^{10}-1.
\end{align*}
\end{example}

\section{Concluding remarks}
Until now, the study of cyclic codes for lengths of the form $\lambda(q^m+1)$ has been confined to the special case $\lambda=1$. This paper extends these results to arbitrary divisors $\lambda$ of $q-1$. For BCH codes, the dimensions of several families are determined, and-more significantly-the minimum distance lower bound for $\mathcal{C}_{(q,n,2\delta+1,n-\delta+1)}$ is raised from $2\delta+1$ to $2(\delta+1)$.
Some of these codes are shown to be optimal. In addition, a necessary and sufficient condition for the dually‑BCH property is established when $m$ is odd. For LCD cyclic codes, the exact count of LCD cyclic codes is given.

Two problems remain and seem to require genuinely new ideas:

The cyclotomic coset structure modulo $\lambda(q^m+1)$ is considerably more complex when $m$ is even. We conjecture that the largest coset leaders are governed by the $2$-adic valuation of $m$. Future work will determine the first few largest coset leaders, investigate the parameters of the resulting BCH codes, and establish necessary and sufficient conditions for them to be dually‑BCH.

The method developed in this paper for treating cyclotomic cosets relies crucially on the condition $\lambda\mid q-1$, which guarantees a certain reversible property.
This approach no longer applies to moduli of the form $(q+1)(q^m\pm1)$, since $q+1\nmid q-1$. Nevertheless, extensive computations suggest that the cyclotomic cosets modulo $(q+1)(q^m\pm1)$ also exhibit favourable properties. Based on extensive computational evidence, we propose the following conjectures:
\begin{itemize}
\item Let $n=(q+1)(q^m+1)$, where $m$ is odd. Then 
\begin{equation*}
\delta_1=	\left\{
	\begin{array}{lcl}
		\frac{3(q+1)(q^m+1)}{4}, ~q\equiv 1\pmod{4};\\
	\frac{(3q+1)(q^m+1)}{4}, ~q\equiv 3\pmod{4}.
	\end{array} \right.
\end{equation*}
and
\begin{equation*}
\delta_2=	\left\{
	\begin{array}{lcl}
		\frac{(3q-1)(q^m+1)}{4}, ~q\equiv 1\pmod{4};\\
	\frac{(3q^2+q-4)q^{m-1}-(q-1)}{4}, ~q\equiv 3\pmod{4}.
	\end{array} \right.
\end{equation*}
\item Let $n=(q+1)(q^m-1)$, where $m$ is odd. Then 
\begin{align*}
\delta_1=q^{m+1}-q^{m-1}-q-1.
\end{align*}
and
\begin{equation*}
\delta_2=	\left\{
	\begin{array}{lcl}
	(q^2-1)q^2-(2q+1), ~m=3;\\
	 (q^2-1)q^{m-1}-(q+1)(q^{\frac{m-1}{2}}+1),~m>3.
	\end{array} \right.
\end{equation*}
\end{itemize}

Then we conjecture that properties of these cyclotomic cosets may yield codes with good parameters. Future work will investigate the cyclotomic coset structure for these lengths and explore the parameters (dimension, minimum distance, weight distribution, etc.) of the corresponding cyclic codes. More broadly, one is led to ask whether codes of length $n=(q^a-1)(q^b+1)$ or $n=(q^a+1)(q^b+1)$ also possess desirable properties.

\section*{Acknowledgment}
The third author was supported by Basic Research Program Young Scientists Guidance Project of Guizhou Province(QN[2025]186).
%
%The authors are grateful to the editor and the anonymous referees for taking time to read and comment on the article very carefully and insightfully. Their nice comments and suggestions helped them to improve the article very much.

\section*{Data availability}
Data sharing not applicable to this article as no datasets were generated or analysed during the current study.

\section*{Declaration of competing interest}
The authors declare that we have no known competing financial interests or personal relationships that 
could have appeared to influence the work reported in this paper.

\end{document}